\documentclass[12pt]{iopart}

\usepackage{graphicx}
\graphicspath{{Materials/}} 
\usepackage{xfrac}
\usepackage{iopams}  
\expandafter\let\csname equation*\endcsname\relax
\expandafter\let\csname endequation*\endcsname\relax
\usepackage{physics}
\usepackage{cleveref}
\usepackage{dsfont}
\usepackage{listings}
\usepackage{xcolor}
\usepackage[section]{placeins}

\begin{document}

\title[Beyond Gaussian Quantum Channels]{Beyond Gaussian Quantum Channels: A model case.}

\author{Daniel Speed, Wenyang Lyu and Roman Schubert}

\address{School of Mathematics, University of Bristol, Bristol BS8 1UG, United Kingdom}
\ead{daniel.speed@bristol.ac.uk, wl16290@bristol.ac.uk, roman.schubert@bristol.ac.uk}
\vspace{10pt}
\begin{indented}
\item{\today}
\end{indented}

\begin{abstract}
Gaussian quantum channels are well understood and have many applications, e.g., in Quantum Information Theory and in Quantum Optics. For more general quantum channels one can in general use semiclassical approximations or perturbation theory, but it is not easy to judge the accuracy of such methods. We study a relatively simple model case, where the quantum channel is generated by a Lindblad equation where one of the Lindblad operators is a multiple of the internal Hamiltonian, and therefore the channel is not Gaussian. For this model we can compute the characteristic function of the action of the channel on a Gaussian state explicitly and we can as well derive a representation of the propagator in an integral form. This allows us to compare the exact results with semiclassical approximations and perturbation theory and evaluate their accuracy. We finally apply these results to the study of the evolution of the von Neumann entropy of a state. 
\end{abstract}

%
%
%
%
%

\section{Introduction}

In quantum mechanics the states of a system are represented by density operators on the Hilbert space $\cH$, i.e.,  positive trace class operators $\hat\rho$ which have trace one. If the system is closed the time evolution of a state is given by conjugation with a unitary time evolution operator $U(t)$, $\hat\rho(t)=U(t)\hat \rho_0U^*(t)$. But if the system is connected to an environment which acts as a source of noise, then the time evolution is described by a family of linear maps $\cE_t$ acting on density matrices $\hat\rho(t)=\cE_t(\hat\rho_0)$. These maps are called quantum channels and they are linear, trace preserving and completely positive.    Quantum channels are used to describe the effect of external operations on a quantum system in combination with the internal dynamics and they provide therefore the natural framework for the theory of Quantum Computation, \cite{nielsen00,BrauvLo05,Hol12}. 

Gaussian quantum channels are a special class of quantum channels on $\R^n$, i.e., in the case where the Hilbert space is infinite dimensional and given by $L^2(\R^n)$. Let us first recall some standard notation. We let $x=(p,q)\in \R^n\times \R^n$ be a phase space point, and $\hat x=(\hat p_1, \cdots, \hat p_n,\hat q_1, \cdots, \hat q_n)$ be the standard position and momentum operators with the commutation relations 
\begin{equation}\label{eq:Omega}
[\hat x_j,\hat x_k]=\ui\hbar \Omega_{jk}\,\, \quad \text{where} \quad \Omega=\begin{pmatrix} 0 & -I\\ I & 0\end{pmatrix}\,\, .
\end{equation}

Here $\Omega$ is the matrix representing the standard symplectic form on $\R^n\times \R^n$ and $I$ denotes the $n\times n$ identity matrix. We can then define the unitary phase space translation operators $T(\xi):=\ue^{-\frac{\ui}{\hbar} \xi\cdot \Omega \hat x}$, where $\xi\in \R^n\times\R^n$, and for any density matrix $\hat \rho$ the corresponding characteristic function is defined as
\begin{equation}
\begin{split}
\chi_{\hat \rho}(\xi):=&\tr[\ue^{\frac{\ui}{\hbar} \xi\cdot \hat x} \hat \rho] \\
=&\int \ue^{\frac{\ui}{\hbar} \xi\cdot x} \rho(x)\,\, \ud x\,\, ,
\end{split}
\end{equation}
where $\rho(x)$  is the Wigner function\footnote{Notice that more commonly the notation $\hat \rho$ denotes that $\rho$ is the Weyl symbol of $\hat\rho$, which differs from the  Wigner function by a factor of  $1/(2\pi\hbar)^n$. In our context it is  more convenient to work with the Wigner function which  is normalised so that $\int \rho(x)\, \ud x =\tr[\hat\rho]$, see as well Appendix A.}  of the state $\hat \rho$. 

One says that $\cE$ is a Gaussian Quantum Channel if there exist a non-negative quadratic form $\xi\cdot D\xi $ and a linear map $R$ such that 
\begin{equation}\label{eq:def-GQC}
\chi_{\cE(\hat \rho)}(\xi)=\ue^{-\frac{1}{2\hbar}\xi\cdot D\xi} \chi_{\hat \rho} (R\xi)\,\, ,
\end{equation}
where the matrices $D$ and $R$ satisfy the relation
\begin{equation}
D+\ui\Omega\geq \ui R^T\Omega R\,\, ,
\end{equation}
which implies the complete positivity of the quantum channel, \cite{BrauvLo05, Hol12, Weed12, AdeRagLee14}. Such channels have the advantage that one can compute many properties of them which are harder to determine for more general quantum channels, and hence have been studied in great detail \cite{holevo1998sending,eisert2005gaussian,AdeRagLee14}.

One way a Gaussian quantum channel can be generated is as the solution to the Lindblad-Gorini–Kossakowski–Sudarshan (LGKS) equation, which describes the evolution of a quantum system which is in contact with an environment, in situations where the evolution is Markovian, i.e, memory effects can be neglected, \cite{AliLen07}. The LGKS equation is of the form 
\begin{equation}
\ui\hbar\pa_t \hat \rho =[\hat H,\hat \rho]+\frac{\ui}{2}\sum_k 2\hat L_k\hat \rho \hat L_k^* -\hat L_k^*\hat L_k \hat \rho-\hat \rho \hat L_k^*\hat L_k\,\, ,
\end{equation}
where $\hat H$ is the internal Hamiltonian and the Lindblad operators $\hat L_k$ describe the coupling to the environment. In the case that $\hat H$ is a quadratic function of $\hat x$ and all the Lindblad operators are linear functions of $\hat x$, i.e., $\hat H=\frac{1}{2}\hat x\cdot Q\hat x$ and $\hat L_k=l_k\cdot \Omega \hat x$, with $Q$ real symmetric and $l_k\in \C^{2n}$, the time evolution is given by a one parameter semigroup of Gaussian quantum channels, \cite{Lind76b,BroOzo10}.  More explicitly for $t\geq 0$ we have $\hat\rho(t)=\cE_t(\hat \rho_0)$, where $\cE_0=I$ and $\cE_t$ is a Gaussian quantum channel with 
\begin{equation}\label{R_and_D_for_GQC}
R_t=\ue^{t(\Omega Q+\Im K\Omega)} \,\, \quad \text{and}\quad D_t=\int_0^t R_s\Re K R_s^T\, \, \ud s\,\, ,
\end{equation}
where $K=\sum_k \bar l_k l_k^T$.

Lindblad operators which are linear in $\hat x$ are frequently used, for instance $\hat L_k= \sigma \hat q_k$ is used to  model the interaction of the system with an environment consisting of random scatterers, \cite{JoosEtAl03,Horn09}. 
Similarly the interaction with a heat bath can be modelled by using suitable multiples of creation and annihilation operators as Lindblad operators \cite{breuer2002theory,AliLen07}. 

Another frequently used choice of Lindblad operators is a multiple of the internal Hamiltonian, 
$\hat L= \gamma \hat H$, this is sometimes called dephasing, see e.g., \cite{wilde13}. 
In this case the corresponding quantum channel will no longer be Gaussian and we will explore  in this paper one particular example to understand how this additional term will affect the characteristics of the quantum channel. 

A state is called a Gaussian state if its characteristic function is a Gaussian. A standard example is a coherent state which is a pure state defined by the wave function
\begin{equation}
u_0^{(g)}(q)=\frac{1}{(\pi \hbar g)^{1/4}} \ue^{-\frac{1}{2\hbar g} q^2}
\end{equation}
where $g>0$, and for simplicity from here on we restricted ourselves to the case $n=1$. The characteristic function and the Wigner function of this state are given by 
\begin{equation}\label{eq:Gaussian-state-diag-G}
\chi_0^{(g)}(\xi)=\ue^{ -\frac{1}{4\hbar} \xi\cdot G\xi}\,\, \text{and} \quad W_0^{(g)}(x)=\frac{1}{\pi\hbar}\ue^{-\frac{1}{\hbar} x\cdot G^{-1}x}\quad\text{with}\quad G=\begin{pmatrix}
g & 0\\ 0 & 1/g  
\end{pmatrix}\,\, ,
\end{equation}
respectively.  
More generally the translated states $u_z^{(g)}:=T(z)u_0^{(g)}$ with $z=(p,q)\in \R^2$ have  characteristic function and Wigner function
\begin{equation}
\chi_z^{(g)}(\xi)=\ue^{\frac{\ui}{\hbar} z\cdot \xi}\chi_0^{(g)}(\xi) \,\, \text{and}\quad 
W_z^{(g)}(x)=W_0^{(g)}(x-z)\,\, , 
\end{equation}
respectively. In many applications superpositions of coherent states $\psi =\frac{1}{\sqrt{\mathcal{N}}}\sum_{j=1}^N u_{z_j}^{(g)}$ 
are very important and their characteristic functions and Wigner functions are given by 
\begin{align}\label{eq:gen-cat-state}
\chi(\xi)&=\frac{1}{\mathcal{N}}\sum_{j,k=1}^N \ue^{\frac{\ui}{2\hbar}z_j\cdot \Omega z_k} \ue^{\frac{\ui}{\hbar}z_{jk}\cdot \xi} \chi_0^{(g)}(\xi-\Omega^T\delta z_{jk})\\
W(x)&=\frac{1}{\mathcal{N}}\sum_{j,k=1}^N \ue^{-\frac{\ui}{2\hbar}z_j\cdot \Omega z_k}\ue^{\frac{\ui}{\hbar} x\cdot \Omega\delta z_{jk}} W_0^{(g)}(x-z_{jk})\,\, ,\label{eq:Wigner-sum}
\end{align}
respectively, where $\mathcal{N}$ is a normalisation constant and 
\begin{equation}
 z_{jk}=\frac{1}{2}(z_j+z_k)\,\, ,\quad \delta z_{jk}=z_j-z_k\,\, .
\end{equation}
We see that the terms with $k=j$ in the Wigner function are positive and correspond to classical probability densities, whereas the terms with $k\neq j$, if $z_j\neq z_k$, give rise to oscillatory terms which are due to quantum interference effects. These are the terms which are suppressed by decoherence, \cite{JoosEtAl03, Horn09}. For Gaussian quantum channels we see by \eqref{eq:def-GQC} that decoherence happens when the matrix $D$ is non-degenerate, because  the terms 
\begin{equation}
\chi_0^{(g)}(R(\xi-\Omega^T\delta z_{j,k}))\ue^{-\frac{1}{2\hbar}\xi\cdot D\xi} =\ue^{-\frac{1}{4\hbar}\big[ (\xi-\Omega^T\delta z_{j,k}) \cdot R^TGR(\xi-\Omega^T\delta z_{j,k})+2 \xi \cdot D\xi\big]}
\end{equation}
are exponentially small in $1/\hbar$ if $\delta z_{j,k}\neq 0$. Therefore if $z_i\neq z_j$ for all $i,j$ applying a Gaussian quantum channel with non-degenerate $D$ to a state \eqref{eq:gen-cat-state} leaves only the diagonal terms with $i=j$ 
\begin{align}
\chi(\xi)&\mapsto_{\mathcal{E}}\frac{1}{\mathcal{N}}\sum_{j=1}^N  \ue^{\frac{\ui}{\hbar} z_j\cdot R\xi-\frac{1}{2\hbar} \xi\cdot D\xi} \chi_0^{(g)}(R\xi)\\
W(x)&\mapsto_{\mathcal{E}}\frac{1}{\mathcal{N}}\sum_{j=1}^N  W_0^{(g')}(x- R^Tz_{j})\,\, , \quad\text{with}\quad G'=R^TGR+2D\,\, ,
\end{align}
up to errors exponentially small in $1/\hbar\times \inf_{j\neq k}\abs{\delta z_{j,k}}^2$, respectively. In particular the Wigner function now looks like a combination of classical probability densities, i.e., decoherence has suppressed the quantum interference terms. 

In the description above we have restricted ourselves to pure states, but we can extend this easily to arbitrary Gaussian states by allowing in \eqref{eq:Gaussian-state-diag-G} a more general $G$, 
\begin{equation}\label{eq:gen-CS}
\chi_0(\xi)=\ue^{ -\frac{1}{4\hbar} \xi\cdot G\xi}\,\, \text{and} \quad W_0(x)=\frac{1}{\pi\hbar\sqrt{\det G}}\ue^{-\frac{1}{\hbar} x\cdot G^{-1}x}\quad\text{with}\quad G+\ui\Omega \geq 0\,\, .  \end{equation}
The condition on $G$ is the Robertson Schr{\"o}dinger uncertainty relation which guarantees that $\chi_0$ is the characteristic function of a density operator. 
There are different conventions about factors of $1/2$ and $1/4$ in the exponents of \eqref{eq:gen-CS}, \cite{AdeRagLee14,Weed12}, in the convention we use the matrix $G$ is related to the covariances of the state via 
\begin{equation}\label{eq:G-cov}
    \frac{\hbar}{2}G=\Gamma:=\begin{pmatrix}\la \hat q^2\ra & \Re \la \hat p\hat q\ra \\ \Re \la \hat p\hat q\ra & \la \hat p^2\ra 
    \end{pmatrix}\,\, ,
\end{equation}
so that a symplectic matrix $G$ corresponds to a state with minimal uncertainties.



\section{The model}

We will consider a free particle, i.e., an internal Hamiltonian $\hat H=\frac{1}{2}\hat{p}^2$ and Lindblad operators 
\begin{equation}\label{eq:two-Lindblads}
\hat L_1=\sigma \, \hat q\,\, \quad \text{and} \quad \hat L_2=\gamma \, \hat H\,\, .
\end{equation}
As both Lindblad operators are Hermitian the  LGKS equation can be rewritten in terms of double commutators, 
\begin{equation}\label{eq:Lindblad-double-commutator}
\ui\hbar \pa_t \hat \rho=[\hat H, \hat \rho]-\frac{\ui}{2}\sum_{j=1}^{2}[\hat L_j[\hat L_j,\hat \rho]]\,\, ,
\end{equation}
which can then be rewritten as an equation for the Wigner function $\rho(p,q)$ of $\hat \rho$ as
\begin{equation}
\pa_t \rho=\{H,\rho\}+\frac{\hbar}{2}\sum_{j=1}^2 \{ L_j,\{L_j,\rho\}\}\,\, . 
\end{equation}
Here we have used the standard result from semiclassical calculus that the Wigner function of a commutator $[\hat A,\hat \rho]$ is $\ui\hbar \{A,\rho\}$ if $\hat A$ is the quantisation of linear or quadratic functions \cite{Zwo12}. 

For the Lindblad operators \eqref{eq:two-Lindblads} the equation becomes
\begin{equation}\label{eq:Lindblad_model-case}
\pa_t\rho=-p\pa_q\rho+\frac{\hbar}{2}\sigma^2 \pa_p^2\rho+\frac{\hbar}{2} \gamma^2 p^2\pa_q^2\rho\,\, 
\end{equation}
where $\rho$ is a function of $(p,q,t)$. We can rewrite this using the vector fields $V_0=p\pa_q$ and $V_1=\pa_p$ as 
\begin{equation}\label{eq:Lindblad-V}
\pa_t\rho=-V_0\rho+\frac{\hbar}{2}\big[\sigma^2 V_1^2\rho+\gamma^2 V_0^2\rho\big]\,\, ,
\end{equation}
so the time evolution of $\rho$ is governed by transport along the vector field $V_0$ and diffusion along the vector fields $V_1$ and $V_0$. 
These vector fields are illustrated in Figure \ref{fig:vectorfields_plot}, we see that in particular at any point $x=(p,q)$ with $p\neq 0$ the two vector fields $V_0$ and $V_1$ span all directions, so diffusion will affect all degrees of freedom. Whereas at $p=0$ the vector field $V_0$ vanishes, and we have diffusion only in the direction of $V_1$, and so a state concentrated near $p=0$ will experience a delayed onset of diffusion in the $q$ direction via the commutator $[V_0,V_1]=-\pa_q$. This is related to the H{\"o}rmander condition in the theory of hypoelliptic equations, \cite{Hor67,Agr20}.

\begin{figure}[t!]
\centering
\includegraphics[width=0.8\linewidth]{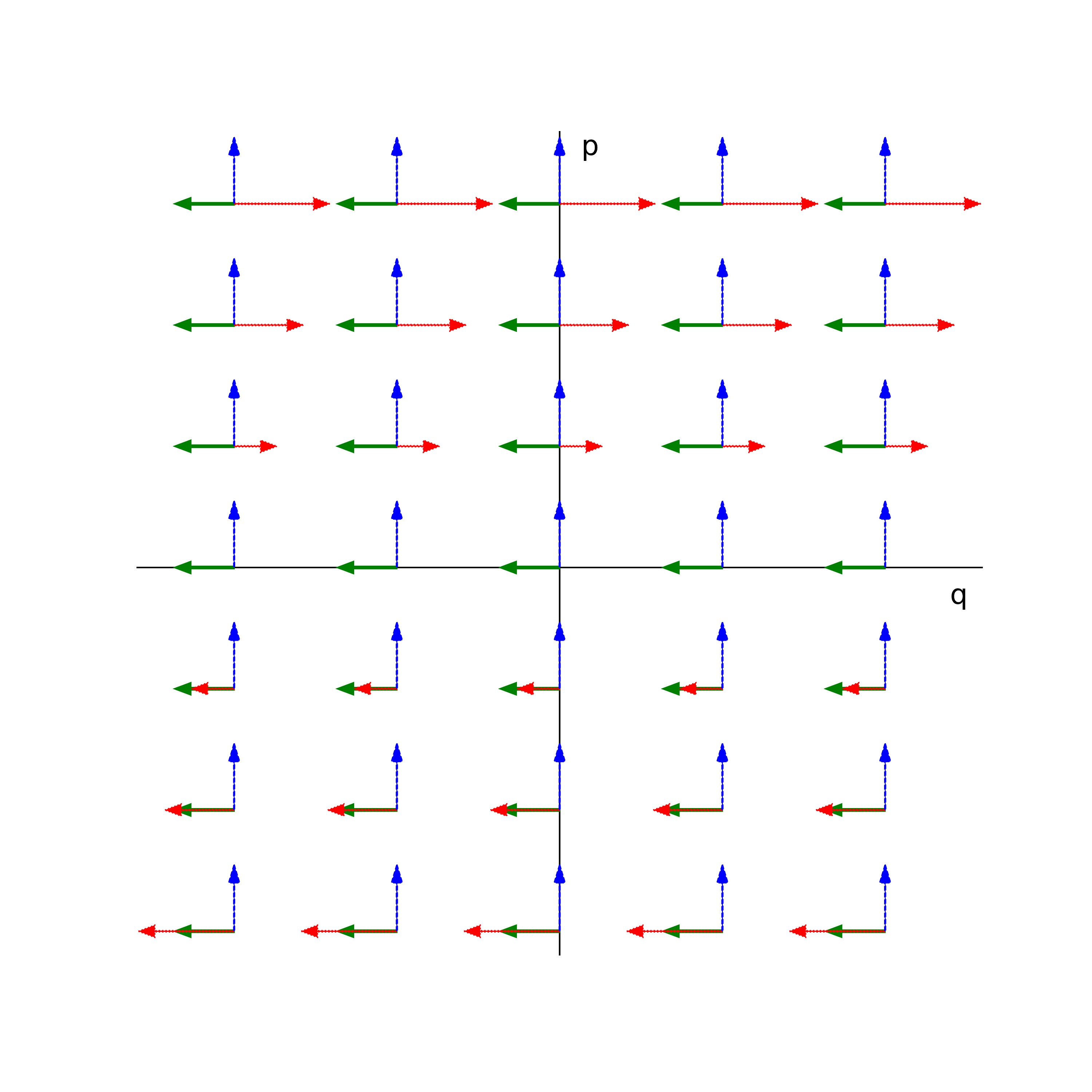}
\caption{The vector fields $V_0$ (red, dotted) and $V_1$ (blue, dashed) which are the building blocks for the phase space LGKS equation \eqref{eq:Lindblad-V}, and their commutator $V_2=[V_0,V_1]$ (green, solid). We see that $V_0$ and $V_1$ span the phase space at all points, except at the line $p=0$ where $V_0$ vanishes, but taking the commutator $V_2$ we still get a spanning set, and therefore decoherence.}\label{fig:vectorfields_plot}
\end{figure}

In the theory of hypoelliptic equations and sub-Riemannian geometry the operator $V_1^2+V_0^2$ is known as the Grushin operator, and \eqref{eq:Lindblad-V} is closely related to the equation for the heat kernel of the Grushin operator, see \cite{Chang09,Barilari12}. 
The strategy we use to solve \eqref{eq:Lindblad-V} is similar to ideas used to study the Grushin operator, see \cite{Chang15}.

In order to solve \eqref{eq:Lindblad_model-case} we introduce a partial Fourier transform of $\rho$, 
\begin{equation}\label{eq:partial-FT}
w(t,p,\eta):=\int \ue^{\frac{\ui}{\hbar} q\eta} \rho(t,p,q)\,\, \ud q
\end{equation}
for which \eqref{eq:Lindblad_model-case} becomes
\begin{equation}\label{eq:w_Lindblad}
\ui \hbar \pa_tw=-\ui\bigg[ -\frac{\hbar^2\sigma^2}{2} \pa_p^2 w+\frac{\gamma^2\eta^2}{2}\, p^2 w\bigg]- \eta p w \,\, .
\end{equation}

This equation is of the form of a one dimensional Schr{\"o}dinger equation, with non-Hermitian Hamiltonian given as the quantisation of the complex Hamilton function 
\begin{equation}\label{eq:mod-Ham}
H(\xi,p)=-\frac{\ui}{2}(\sigma^2\xi^2+\gamma^2\eta^2 p^2)-\eta p
\end{equation}
with phase space variable $(\xi,p)$. This is a second order polynomial in $(\xi,p)$ and hence we can solve the Schr{\"o}dinger equation explicitly. 

We will first look at the special case of propagation of Gaussian wave packets and then discuss the general solutions.

\subsection{Gaussian Wavepackets}\label{sec:Gaussian-prop}

We consider the case that the initial state has a Wigner function that is a sum of Gaussians of the form 
\begin{equation}\label{eq:Gaussian-0}
\rho_0(p,q)=\ue^{\frac{\ui}{\hbar}x\cdot \Omega\delta z}W_0^{(g)}(x-z_0)
=\frac{1}{\pi\hbar}\ue^{\frac{\ui}{\hbar}(q\delta p-p\delta q)} \ue^{-\frac{1}{\hbar} [g (q- q_0)^2+\frac{1}{g} (p- p_0)^2]}\,\, ,
\end{equation}
where $x=(p,q)$, and $z_0=(p_0,q_0)$ and $\delta z=(\delta p,\delta q)$ are parameters. Recall that \eqref{eq:Wigner-sum} is of this form, in particular of the initial state is a cat state, i.e., a superposition of two coherent states, we have by \eqref{eq:Wigner-sum} 
\begin{equation}\label{eq:cat-state}
    \rho^{(0)}(x)=\rho^{(0)}_{d}(x)+\rho^{(0)}_{osc}(x)
\end{equation}
with 
\begin{align}
 \rho^{(0)}_{d}(x)&=\frac{1}{\mathcal{N}\pi\hbar} \ue^{-\frac{1}{\hbar} [g (q- q_1)^2+\frac{1}{g} (p-p_1)^2]}+\frac{1}{\mathcal{N}\pi\hbar} \ue^{-\frac{1}{\hbar} [g (q-q_2)^2+\frac{1}{g} (p-p_2)^2]}\label{eq:rho-d}\\
 \rho^{(0)}_{osc}(x)&=\frac{1}{\mathcal{N}\pi\hbar}\ue^{\frac{\ui}{\hbar}(q\delta p-p\delta q)} \ue^{-\frac{1}{\hbar} [g (q-q_0)^2+\frac{1}{g} (p- p_0)^2]}+\frac{1}{\mathcal{N}\pi\hbar}\ue^{-\frac{\ui}{\hbar}(q\delta p-p\delta q)} \ue^{-\frac{1}{\hbar} [g (q- q_0)^2+\frac{1}{g} (p- p_0)^2]}\,\, ,\label{eq:rho-osc}
\end{align}
where $p_0=(p_1+p_2)/2$, $q_0=(q_1+q_2)/2$ and $\delta p=p_2-p_1$, $\delta q=q_2-q_1$, which is clearly a sum of terms of the form \eqref{eq:Gaussian-0}. Here the normalisation constant is given by 
\begin{equation}
    \mathcal{N}=2+2\cos\bigg(\frac{1}{\hbar}(q_0\delta p-p_0\delta q)\bigg)\ue^{-\frac{1}{4\hbar}[(\delta p)^2+(\delta q)^2]}.
\end{equation}
In Figure \ref{fig:rho_evolution3d_cat} we present a plot of $\rho^{(0)}$ and its time evolution. We see clearly the differences in the evolution of $\rho_d$ and $\rho_{osc}$, in particular that the oscillatory part is dampened out extremely rapidly, which is the effect of decoherence.

\begin{figure}[t!]
\centering
\subfigure[$\gamma=1,t=0$]{\includegraphics[width=0.3\linewidth]{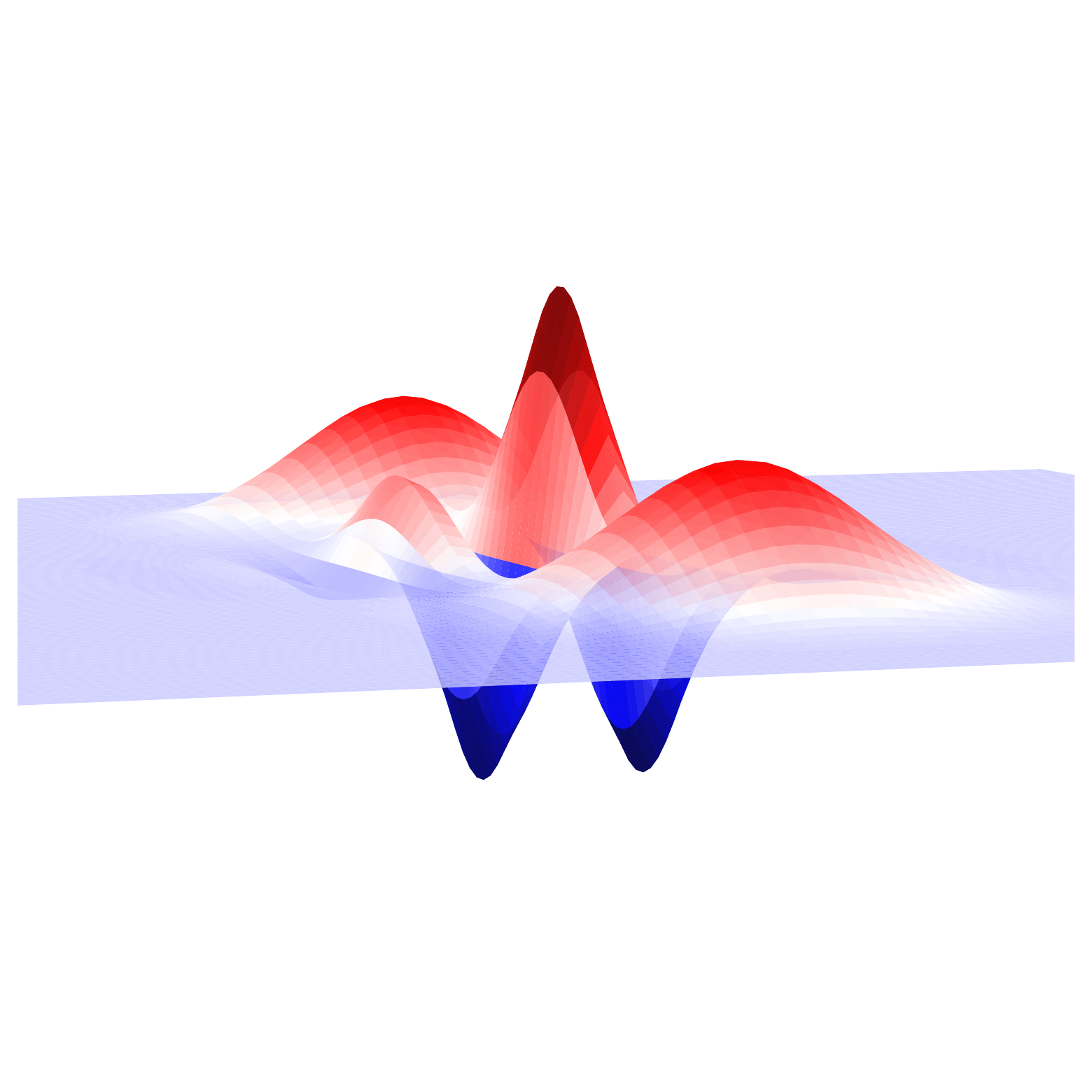}}
\subfigure[$\gamma=1,t=0.05$]{\includegraphics[width=0.3\linewidth]{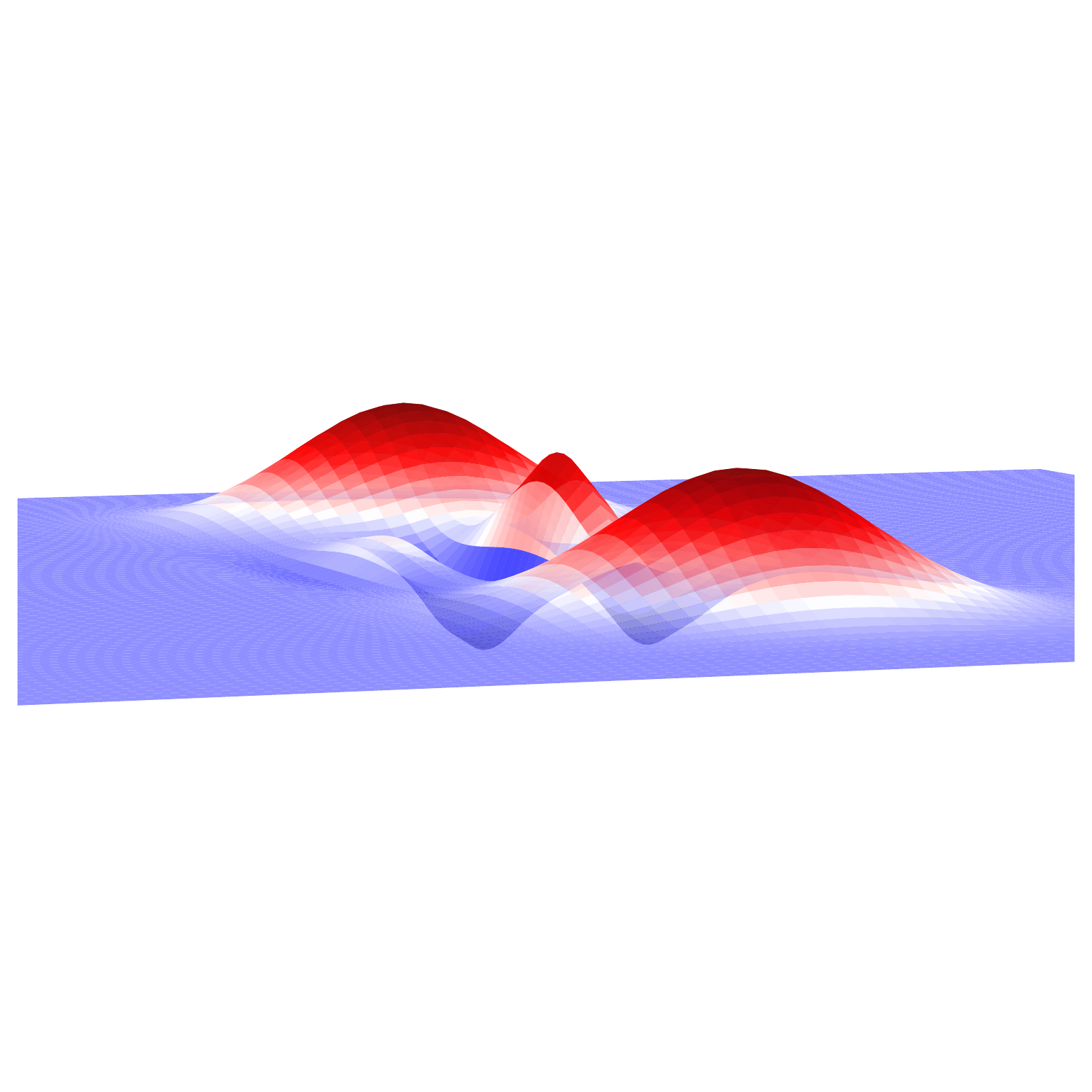}}
\subfigure[$\gamma=1,t=0.1$]{\includegraphics[width=0.3\linewidth]{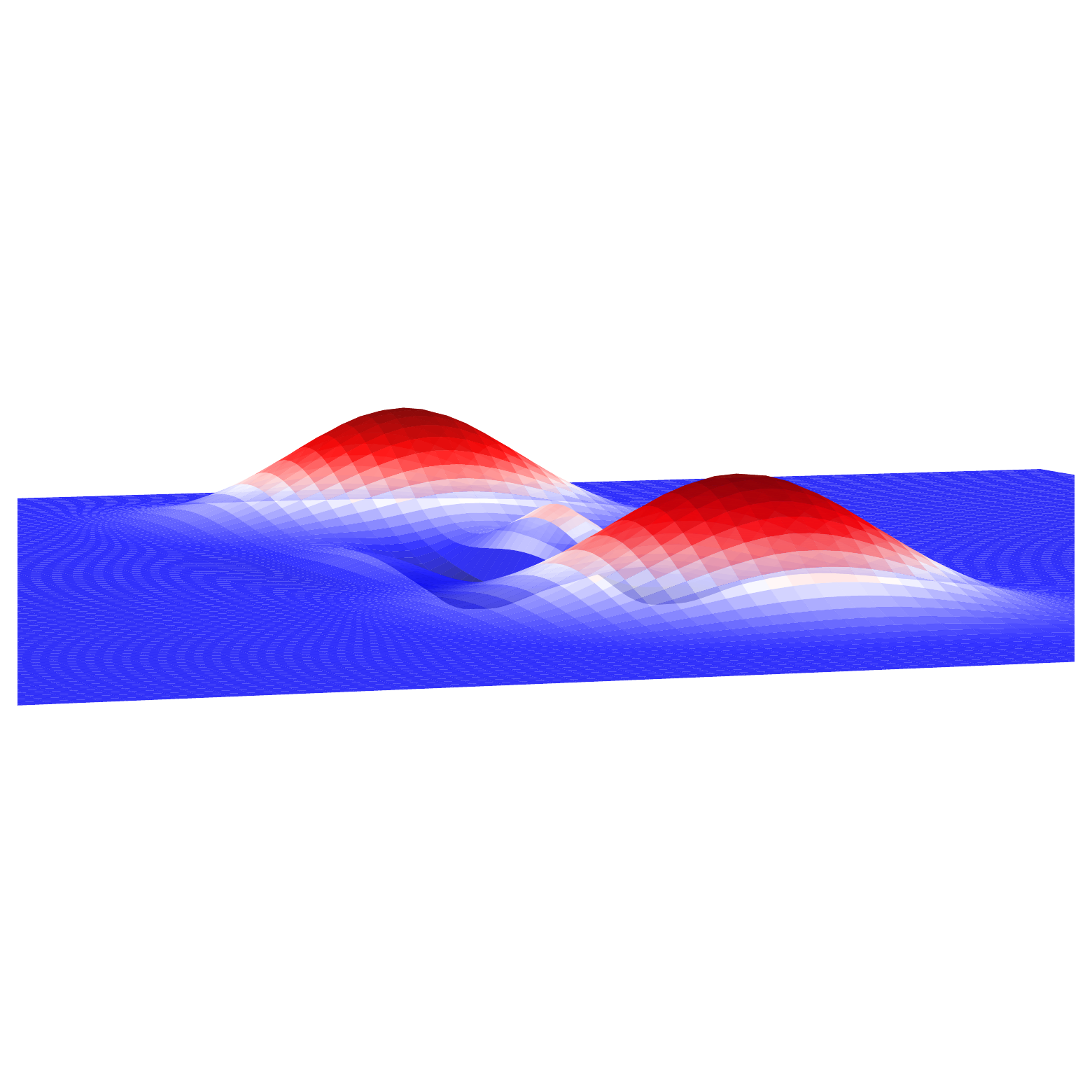}} \\
\caption{Evolution of $\rho(t)$ with initial condition given by the cat state \eqref{eq:cat-state} with $z_1=(p_1,q_1)=(0,3)$, $z_2=(p_2,q_2)=-z_1$,  $g=1$ and for $\hbar=1$, $\sigma=1$ and $\gamma=1$.  We see the characteristic two Gaussians \eqref{eq:rho-d} corresponding to the individual coherent states at $z_1$ and $z_2$ and the oscillating term \eqref{eq:rho-osc}  in the centre $z_0=(z_1+z_2)/2$ which is damped away rapidly due to the diffusive terms in the GKLS equation \eqref{eq:Lindblad-V}, a manifestation of decoherence. }\label{fig:rho_evolution3d_cat}
\end{figure}

In order to simplify the computations we have furthermore assumed that $G$ is diagonal, as in \eqref{eq:Gaussian-state-diag-G}. Now the partial Fourier transform of \eqref{eq:Gaussian-0} in $q$, see \eqref{eq:partial-FT}, gives 
\begin{equation}
w_0(p,\eta)=\ue^{\frac{\ui}{\hbar} q_0(\eta+\delta p)} \ue^{-\frac{1}{4\hbar g}(\eta +\delta p)^2}\frac{1}{\sqrt{g\pi\hbar}} \ue^{-\frac{\ui}{\hbar}p\delta q} \ue^{-\frac{1}{\hbar} \frac{1}{g} (p- p_0)^2}
\end{equation}
and following the ideas in \cite{GraeSchu11,GraeSchu12} we make an Ansatz 
\begin{equation}
w(t,p,\eta)=\ue^{\frac{\ui}{\hbar}q_0(\eta+\delta p)} \ue^{-\frac{1}{4\hbar g}(\eta +\delta p)^2}\frac{c(t)}{\sqrt{g\pi\hbar}} \ue^{-\frac{1}{\hbar} D(t)+\frac{\ui}{\hbar}\phi(t)}\ue^{-\frac{\ui}{\hbar}pQ(t)} \ue^{-\frac{1}{\hbar} \frac{1}{a(t)} (p-P(t))^2}
\end{equation}
where the parameters $c(t), D(t),\phi(t), a(t), P(t), Q(t)$ can depend on $\eta$ but not on $p$. Inserting this Ansatz into \eqref{eq:w_Lindblad} gives the following set of equations for the parameters
\begin{align}
\dot a&=2\sigma^2-\frac{\gamma^2\eta^2}{2} a^2&\text{with}\quad a(0)&=g\,\, ,\\
\dot P&=-\frac{\gamma^2\eta^2}{2} a P &\text{with}\quad P(0)&= p_0\,\, ,\\
\dot Q&=-\frac{2\sigma^2}{a}Q-\eta &\text{with}\quad Q(0)&=\delta q\,\, ,\\
\dot c&=-\frac{\sigma^2}{a} c &\text{with}\quad c(0)&=1\,\, , \label{eq:c0}\\
\dot D&=\frac{\sigma^2}{2}Q^2+\frac{\gamma^2\eta^2}{2}P^2 &\text{with}\quad D(0)&=0\,\, ,\label{eq:D}\\
\dot \phi &=\dot Q P+\eta P &\text{with}\quad \phi(0)&=0\label{eq:phi}\,\, ,
\end{align}
as we show in Appendix \ref{app:wpp}. These equations can be solved explicitly, see Appendix \ref{app:wpp} for the details, and we therefore obtain an explicit solution for $w_t(p,\eta)$ if the initial condition is Gaussian. It is convenient to express the solutions in terms of two auxiliary functions 
\begin{equation}
u(t,\omega):=\cosh(\omega t)+\frac{\omega}{\beta}\sinh(\omega t)\,\, ,\quad v(t,\omega):=\cosh(\omega t)+\frac{\beta}{\omega}\sinh(\omega t)
\end{equation}
where
\begin{equation}
\omega=\sigma\gamma\eta\,\, ,\quad\text{and}\quad \beta =\frac{2\sigma^2}{g}\,\, ,
\end{equation}
and we furthermore will use the notation
\begin{equation}
\ch (t,\omega):=\frac{\cosh(t\omega)-1}{\omega^2}\,\, \quad\text{and}\quad \sh(t,\omega):=\frac{\sinh(t\omega)}{\omega}\,\, .
\end{equation}
Notice that these functions are smooth at $\omega=0$, and we use them to make sure that we do not create apparent singularities at $\omega=0$ by careless notation. 
The solutions to the set of differential equations are then given by 
\begin{equation}
a(t)=g\frac{v(t, \omega)}{u(t, \omega)}\,\, ,\quad c(t)=\frac{1}{\sqrt{v(t,\omega)}} 
\end{equation}
and 
\begin{equation}\label{eq:P-Q}
P(\eta,t)=\frac{ p_0}{u(t,\omega)}\,\, ,\quad Q(\eta,t)=\frac{\delta q-\eta [\beta \ch(t,\omega)+\sh(t,\omega)]}{v(t,\omega)}\,\, ,
\end{equation}
and the remaining two terms are given by 
\begin{equation}
\phi(t,\eta)=\beta  p_0\frac{\eta \ch(t,\omega)-\delta q \sh(t,\omega)}{v(t,\omega)}
\end{equation}
and 
\begin{equation}
\begin{split}
D(t,\eta)&= \frac{1}{2\gamma^2}\frac{t v(t,\omega)-\sh(t,\omega)-2\beta \ch(t,\omega)}{v(t,\omega)}\\
&+\frac{\sigma^2}{2}\frac{-2\eta \delta q \sh(t, \omega)+(\delta q)^2\ch(t,\omega)}{v(t,\omega)}+\frac{\gamma^2\eta^2 p_0^2}{2}\frac{\sh(t,\omega)}{u(t,\omega)} .
\end{split}
\end{equation}

Furthermore, we can compute the Fourier transform in $p$, $\chi(t,\xi,\eta)=\int \ue^{\frac{\ui}{\hbar} p\xi} w_t(p,\eta)\,\, \ud p$, as $w_t(p,\eta)$ is Gaussian in $p$, and therefore obtain an explicit expression for the characteristic function
\begin{equation}\label{eq:char_function}
\begin{split}
\chi(t,\xi,\eta)&= c(t)\sqrt{a(t)/g}\,\, \ue^{-\frac{1}{\hbar} D(t)+\frac{\ui}{\hbar}\phi(t)}\ue^{\frac{\ui}{\hbar} q_0(\eta+\delta p)+\frac{\ui}{\hbar} P(\xi-Q)} \ue^{-\frac{1}{4\hbar g}(\eta +\delta p)^2-\frac{a}{4\hbar}(\xi-Q)^2}\\
&=\frac{1}{\sqrt{u(t,\omega)}}\ue^{-\frac{1}{\hbar} D(t,\eta)+\frac{\ui}{\hbar}\phi(t,\eta)}\ue^{\frac{\ui}{\hbar} q_0(\eta+\delta p)+\frac{\ui}{\hbar} P(\eta ,t)(\xi-Q(\eta,t))} \ue^{-\frac{1}{4\hbar g}(\eta +\delta p)^2-\frac{a}{4\hbar}(\xi-Q(\eta, t))^2}\,\, .
\end{split}
\end{equation}
This expression for the characteristic function, together with the explicit expressions for the parameters we gave above, is one of our main results.

We will discuss now some of the properties of the solutions we found. 
Let us first consider the case that $\delta p=\delta q=0$, then the Wigner function \eqref{eq:Gaussian-0} and the time evolved characteristic function \eqref{eq:char_function} represent a positive state, and we can use the characteristic function to study properties of this state. In particular, moments can be computed from derivatives of $\chi(t,\xi,\eta)$ at 
$\xi=\eta=0$, i.e., 
\begin{equation}
    \la \hat p^n\hat q^m\ra_t:=\tr\big[\hat p^n\hat q^m\hat \rho(t)\big]=\big(-\ui\hbar)^{n+m}\big[\pa_{\xi}^n\pa_{\eta}^m\chi\big](t,0,0)\,\, .
\end{equation}
Using this relation and \eqref{eq:char_function}
we find
\begin{equation}\label{eq:first-moments}
    \la \hat p\ra_t=p_0\,\, ,\quad \la \hat q\ra_t=q_0+tp_0
\end{equation}
so the momentum and position expectation values follow the internal dynamics of the system. But for the variances we obtain
\begin{align}
    \la \hat p^2\ra_t-\la \hat p\ra_t^2&=\hbar\bigg(\frac{g}{2}+\sigma^2 t\bigg)\label{eq:full-variance-p}\\
 \frac{1}{2}\la \hat p\hat q+\hat q\hat p\ra_t-\la \hat p\ra_t\la \hat q\ra_t&=\hbar\bigg(\frac{g}{2}t+\frac{\sigma^2}{2} t^2\bigg)\label{eq:full-variance-pq}\\  
 \la \hat q^2\ra_t-\la \hat q\ra_t^2&=\hbar\bigg(\frac{1}{2g}+\gamma^2 p_0^2t+\frac{g}{2}t^2+\frac{\sigma^2}{3}t^3\bigg)+\hbar^2\frac{\gamma^2}{2}(g t+\sigma^2 t^2)\label{eq:full-variance-q}
\end{align}
and we see that as expected they depend on the Lindblad operators \eqref{eq:two-Lindblads} as can be seen by the parameters $\sigma^2$ and $\gamma^2$. How the variances depend on the coupling to the environment can be understood in terms of the corresponding vector fields $V_0$ and $V_1$, as depicted in Figure \ref{fig:vectorfields_plot}. Since the vector field $V_1$ is constant the parts in the variances proportional to $\sigma^2$ do not depend on where the initial state is concentrated. But the vector  field $V_0$ depends on $p$ and we see that the corresponding contribution to the position variance depends on the initial momentum $p_0$, and furthermore the variance for $\hat q$ has an $\hbar^2$ term. If $p_0=0$ the order $\hbar$ terms in the variance have no contribution from $V_0$, since $V_0$ vanishes at $p=0$, and the contribution becomes only visible in order $\hbar^2$. We will see in Section \ref{sec:semi-approx} that the standard semiclassical approximation does not detect this higher order contribution. 

For $\gamma\to 0$ the state converges to a Gaussian state with covariances given by \eqref{eq:full-variance-p}, \eqref{eq:full-variance-pq} and \eqref{eq:full-variance-q}, which together with the first moments \eqref{eq:first-moments} determines the state uniquely.

In Figure \ref{fig:rho_evolution_p0_0} we give the evolution of an initial Gaussian state under the Lindblad evolution for different values of $\gamma$. For $\gamma=0$ the state stays Gaussian, but the variance evolves. For $\gamma=0.5$ and $\gamma=1$ the state still stays localised but we see that it slowly develops non-Gaussian features, in particular in the tails.

\begin{figure}[t!]
\centering
\subfigure[$\gamma=0,t=0$]{\includegraphics[width=0.3\linewidth]{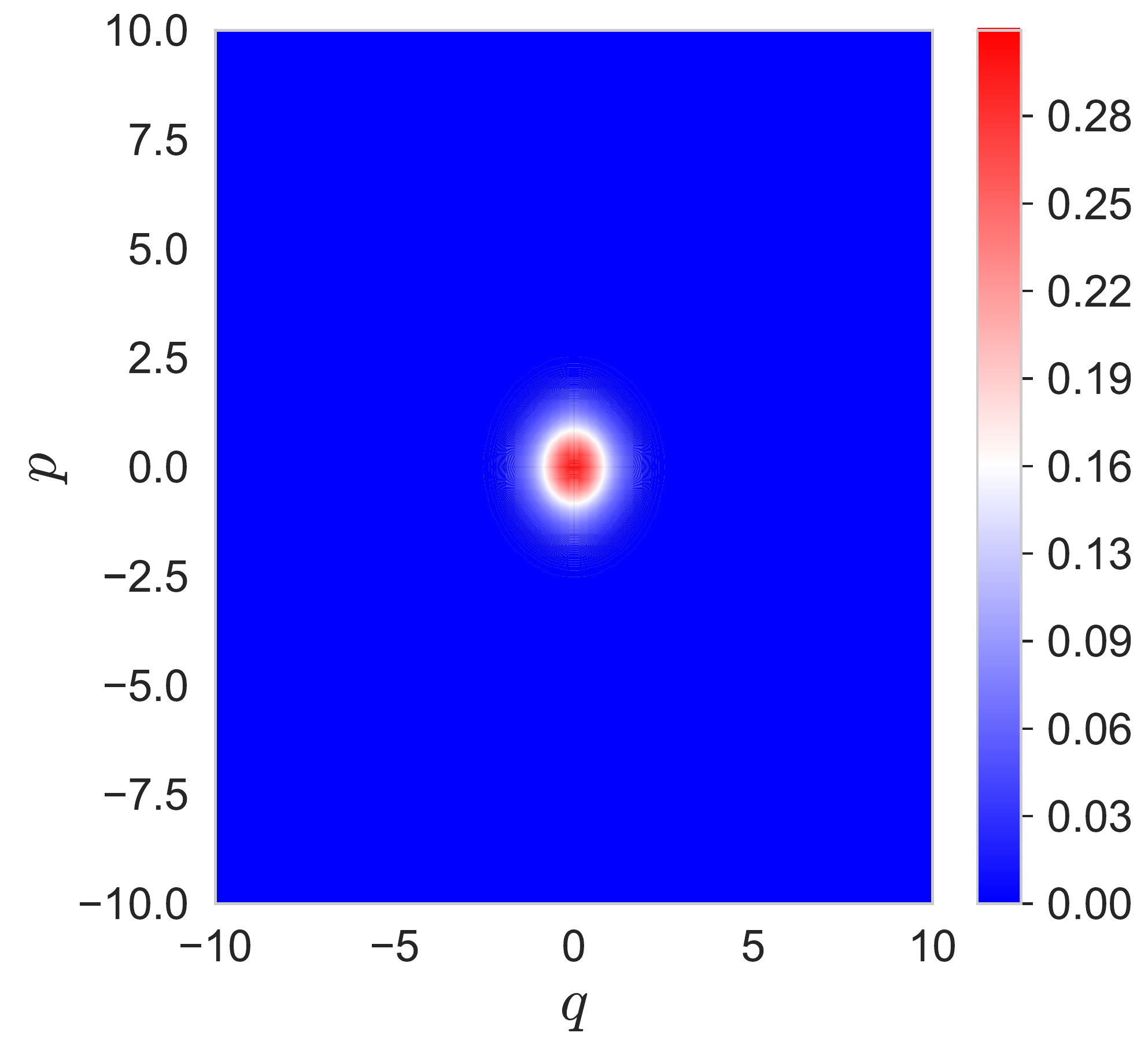}}
\subfigure[$\gamma=0,t=0.5$]{\includegraphics[width=0.3\linewidth]{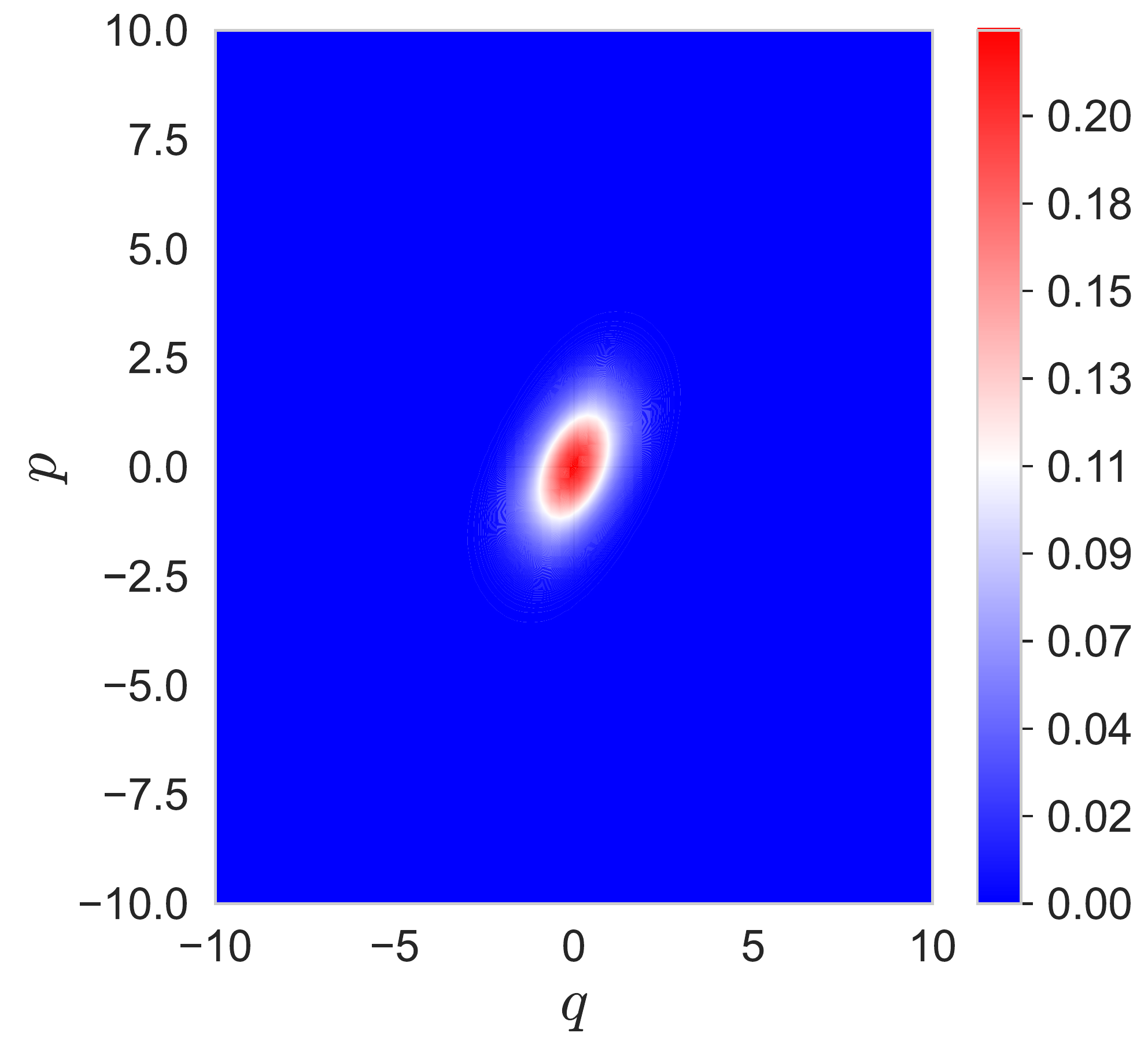}}
\subfigure[$\gamma=0,t=1$]{\includegraphics[width=0.3\linewidth]{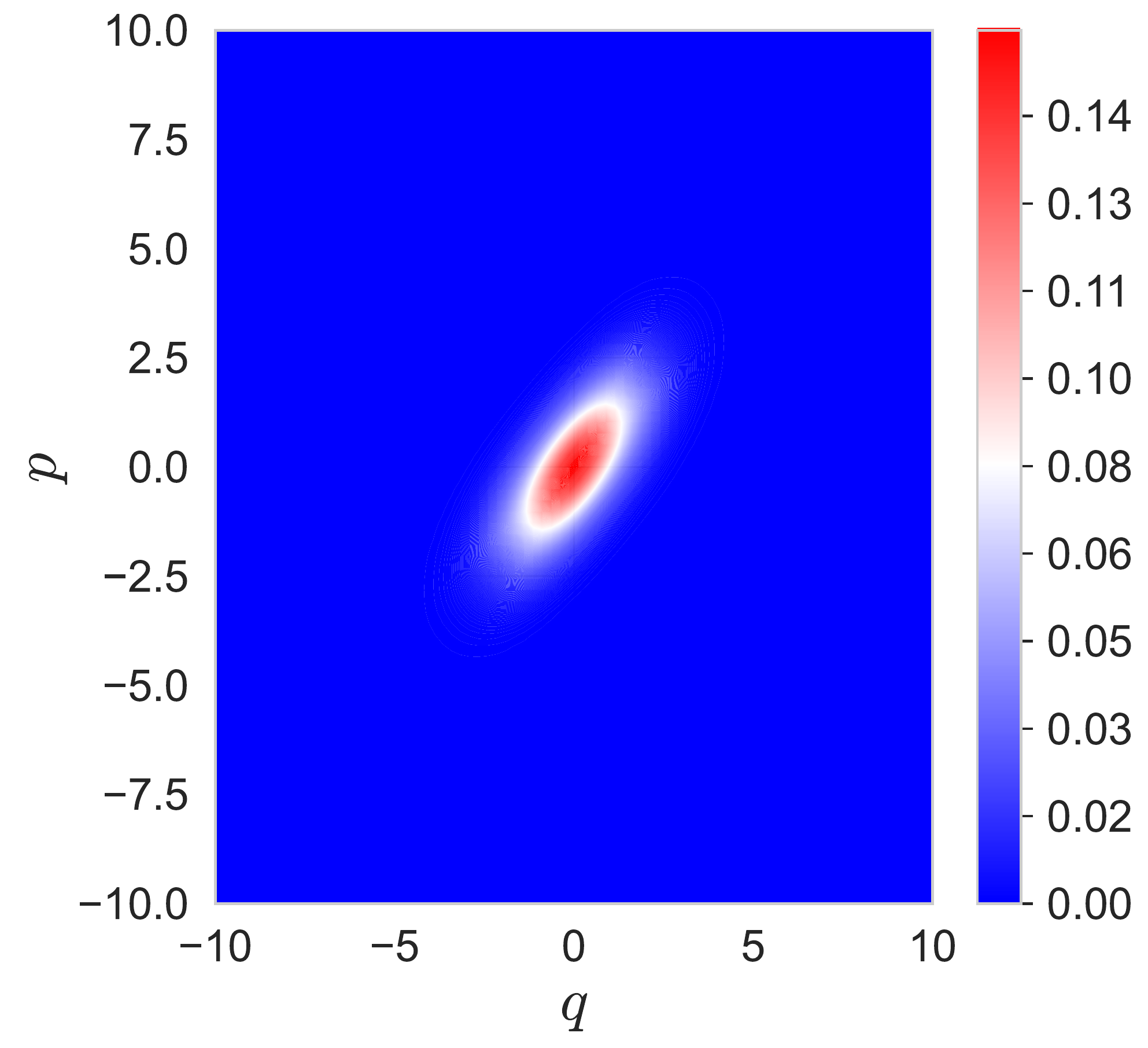}} \\
\subfigure[$\gamma=0.5,t=0$]{\includegraphics[width=0.3\linewidth]{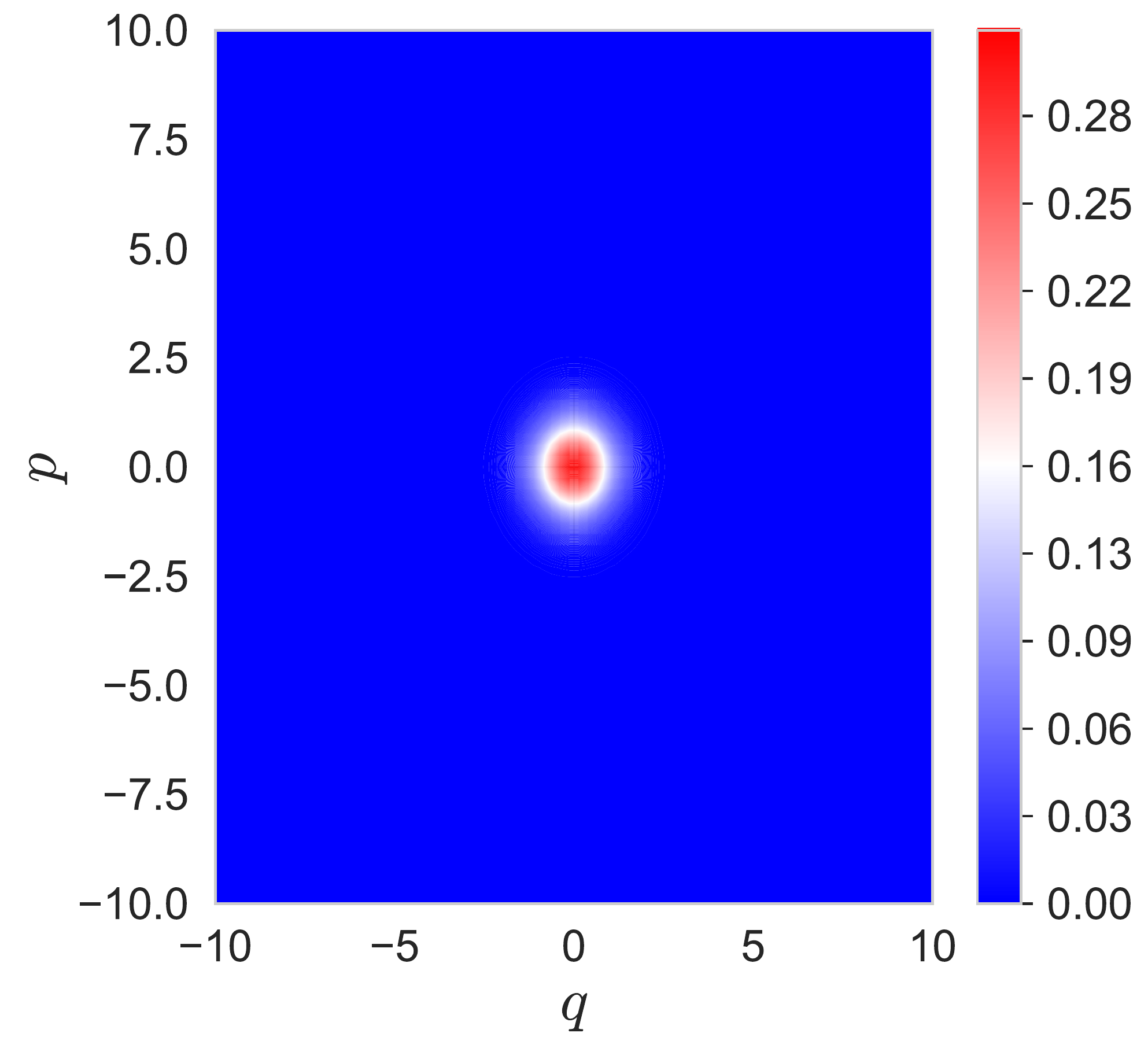}}
\subfigure[$\gamma=0.5,t=0.5$]{\includegraphics[width=0.3\linewidth]{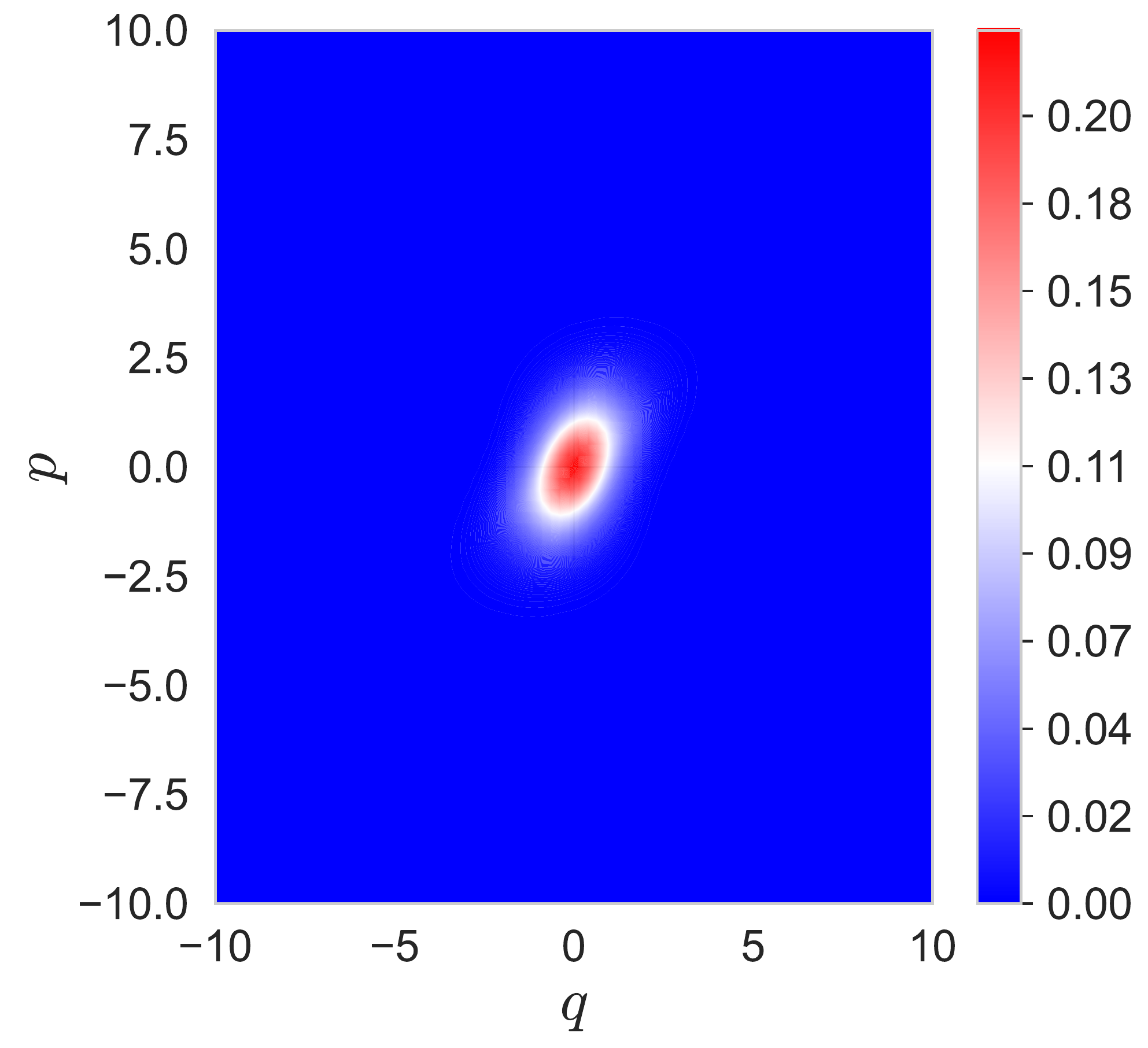}}
\subfigure[$\gamma=0.5,t=1$]{\includegraphics[width=0.3\linewidth]{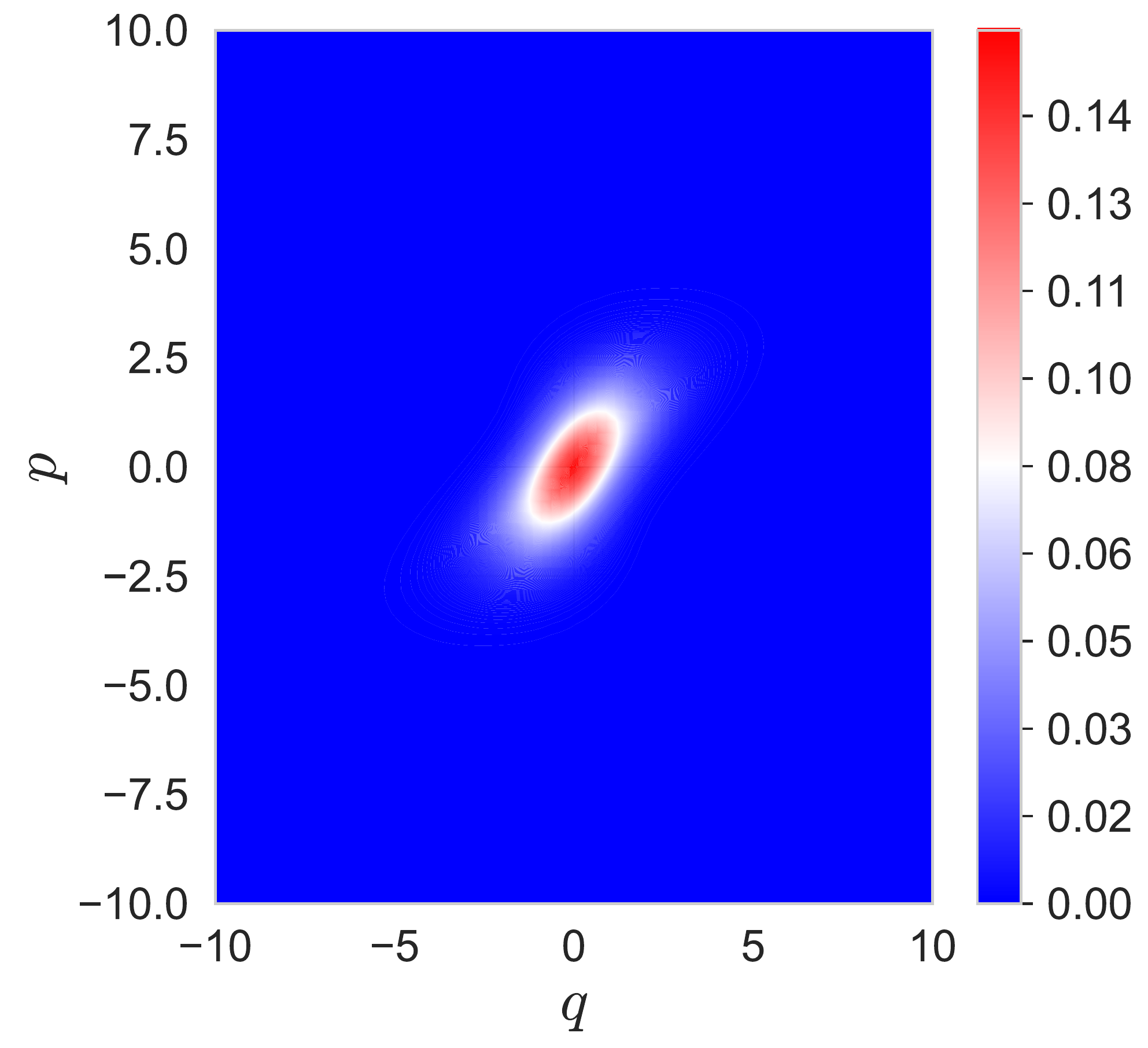}} 
\\
\subfigure[$\gamma=1,t=0$]{\includegraphics[width=0.3\linewidth]{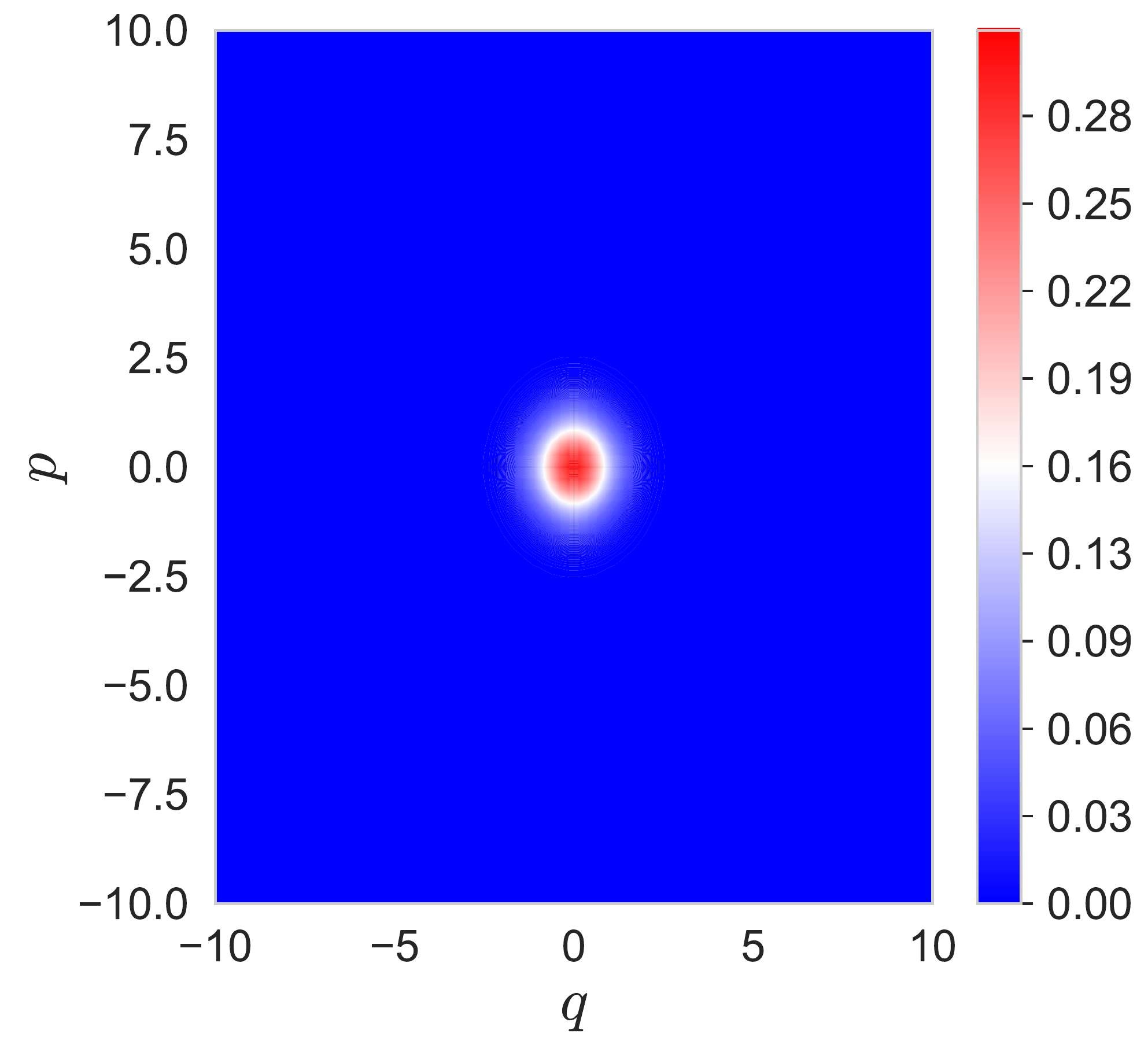}}
\subfigure[$\gamma=1,t=0.5$]{\includegraphics[width=0.3\linewidth]{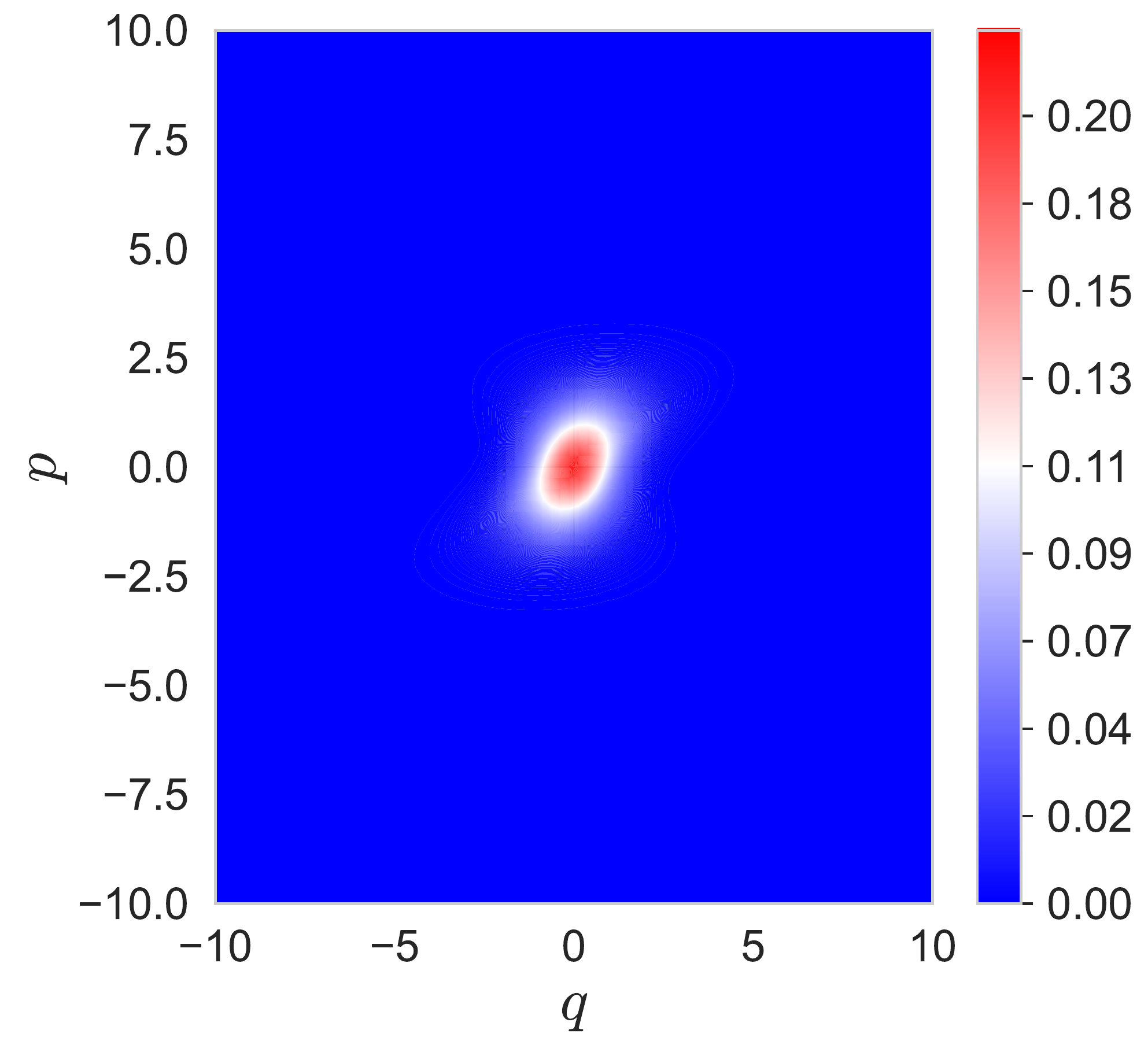}}
\subfigure[$\gamma=1,t=1$]{\includegraphics[width=0.3\linewidth]{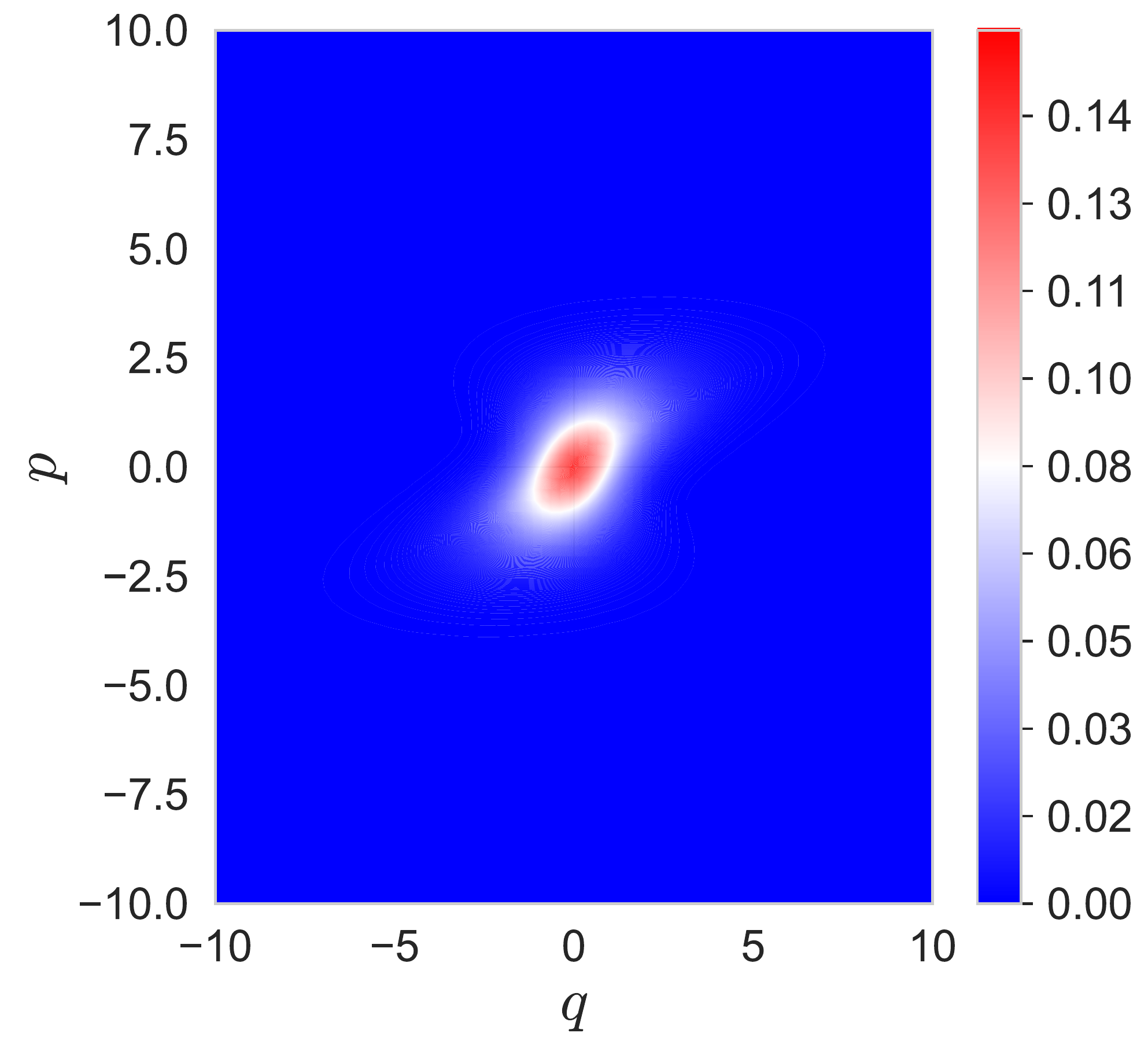}}
\caption{Evolution of an initially Gaussian state $\rho$ of the form \eqref{eq:Gaussian-0} with $q_0=p_0=0$, $g=0$ and $\hbar=1$, for $\sigma=1$ and three different $\gamma$ values. Top panel: $\gamma=0$, the evolution is Gaussian, the centre is not moving but the variance is evolving. Middle panel: $\gamma=0.5$ , Bottom panel: $\gamma=1$ the state stays localised but starts deviating from a Gaussian.}\label{fig:rho_evolution_p0_0}
\end{figure}

Let us now consider  the case that $(\delta p, \delta q)\neq 0$. In this situation we expect decoherence to cause $\chi(t)$ to be suppressed rapidly. We will quantify this decay by computing $\tr[\hat \rho^*\hat\rho]$.

Recall that $\tr[\hat \rho^*\hat \rho]=\frac{1}{(2\pi\hbar)^2}\iint \abs{\chi(t, \xi,\eta)}^2\, \ud \xi\ud \eta$, where we can insert the expression \eqref{eq:char_function} and evaluate the $\xi$-integral to obtain
\begin{equation}
\tr [\hat \rho^*\hat \rho]=\frac{1}{(2\pi\hbar)^{3/2}}\int \frac{1}{\sqrt{gv(t,\eta)}}\ue^{-\frac{1}{\hbar}[2D(t, \eta)+(\eta-\delta p)^2/(2g)]}\, \ud \eta\,\, .
\end{equation}
From the differential equation for $D$, \eqref{eq:D}, we see that $D(t, \eta)\geq 0$ and that $D(t, \eta)$ is non-decreasing as a function of $t$, hence if we can establish a lower bound for small $t$ it will as well be valid for larger times. We have as well the Gaussian factor $\ue^{-\frac{1}{4g \hbar}(\eta-\delta p)^2}$ in the expression \eqref{eq:char_function} for the characteristic function, and therefore $\chi(t, \xi, \eta)$ will be small away from a neighbourhood of $\eta=\delta p$, in particular if $\hbar$ is small.

To understand the onset of decoherence let us determine the leading order terms in the Taylor expansion of $D$ in $t$, where from \eqref{eq:P-Q} we find 
\begin{equation}
P(t)= p_0(1+O(t))\,\, \quad \text{and} \quad Q(t)=\frac{\delta q+\eta t+O(t^2)}{1+\beta t}=\delta q+(\eta -\beta\delta q)t+O(t^2) .
\end{equation}
Inserting this into \eqref{eq:D} we obtain
\begin{equation}\label{eq:D-short-time}
D(t, \eta)=\begin{cases}\big[\frac{\sigma^2}{2}(\delta q)^2+\frac{\gamma^2\eta^2}{2}p_0^2\big]t +O(t^2)& (\delta q, p_0)\neq 0\\
\frac{\sigma^2\eta^2}{6}\, t^3 +O(t^4) & \delta q= p_0=0\end{cases}\,\, .
\end{equation}
We observe here two different regimes for the onset of decoherence, we have a linear growth of $D(t)$ in $t$ for most cases, but if $p_0=0$ and $\delta q=0$, but $\delta p\neq 0$, the onset of decoherence is slower for small times, it is proportional to $t^3$ instead of $t$. This can be understood in terms of the geometry of the vector fields which describe diffusion and transport (Figure \ref{fig:vectorfields_plot}). If $ p_0\neq 0$, then there are two linearly independent vector fields describing diffusion, and hence the onset  of decoherence is immediate. But if $ p_0=0$ and $\delta q=0$, then one of the vector fields vanishes, and the other is orthogonal to the oscillations of the Wignerfunction, and only their commutator couples to the oscillations, hence the onset of decoherence is delayed. This is related to the H{\"o}rmander condition for hypoelliptic operators, \cite{Hor67,Agr20}. We see the same phenomenon in the Gaussian channel case if $\gamma=0$, the only difference if $\gamma>0$ is that if $ p_0\neq 0$ then  decoherence is enhanced by the presence of the additional Lindblad term. 

For the short time approximation we can assume that $v(t,\eta)=v(0,\eta)=1$ and then we can evaluate the integral if we approximate $D(t,\eta)$ by the leading order terms in $t$ in \eqref{eq:D-short-time} and obtain
\begin{equation}\label{small_time_trace_rho}
    \frac{\tr[\hat\rho^*(t)\hat\rho(t)]}{\tr[\hat \rho_0^*\hat\rho_0]}=\begin{cases} \ue^{-\frac{1}{\hbar}\,t [\sigma^2(\delta q)^2 +\gamma^2  p_0^2(\delta p^2]+O(t^2)}(1+O(t)) & (\delta q,\bar p)\neq 0\,\, ,\\
    \ue^{-\frac{1}{\hbar}\, t^3\frac{\sigma^2}{3}(\delta p)^2+O(t^4)}(1+O(t)) &\delta q=\bar p=0\,\, .
    \end{cases}
\end{equation}
Notice that we have for self-adjoint Lindblad operators that 
\begin{equation}
 \frac{\ud  \tr[\hat\rho^*(t)\hat\rho(t)]}{\ud t}=-\sum_j  \tr\big[[L_j,\rho(t)]^*[L_j,\rho(t)]\big]\leq 0\,\, ,
\end{equation}
see \cite{AliLen07}, and therefore the short time estimate in \eqref{small_time_trace_rho} implies that  $\tr[\hat\rho^*(t)\hat\rho(t)]$ remains exponentially small in $1/\hbar$ for large times, too.

\begin{figure}[t!]
\begin{center}\includegraphics[width=0.7\linewidth]{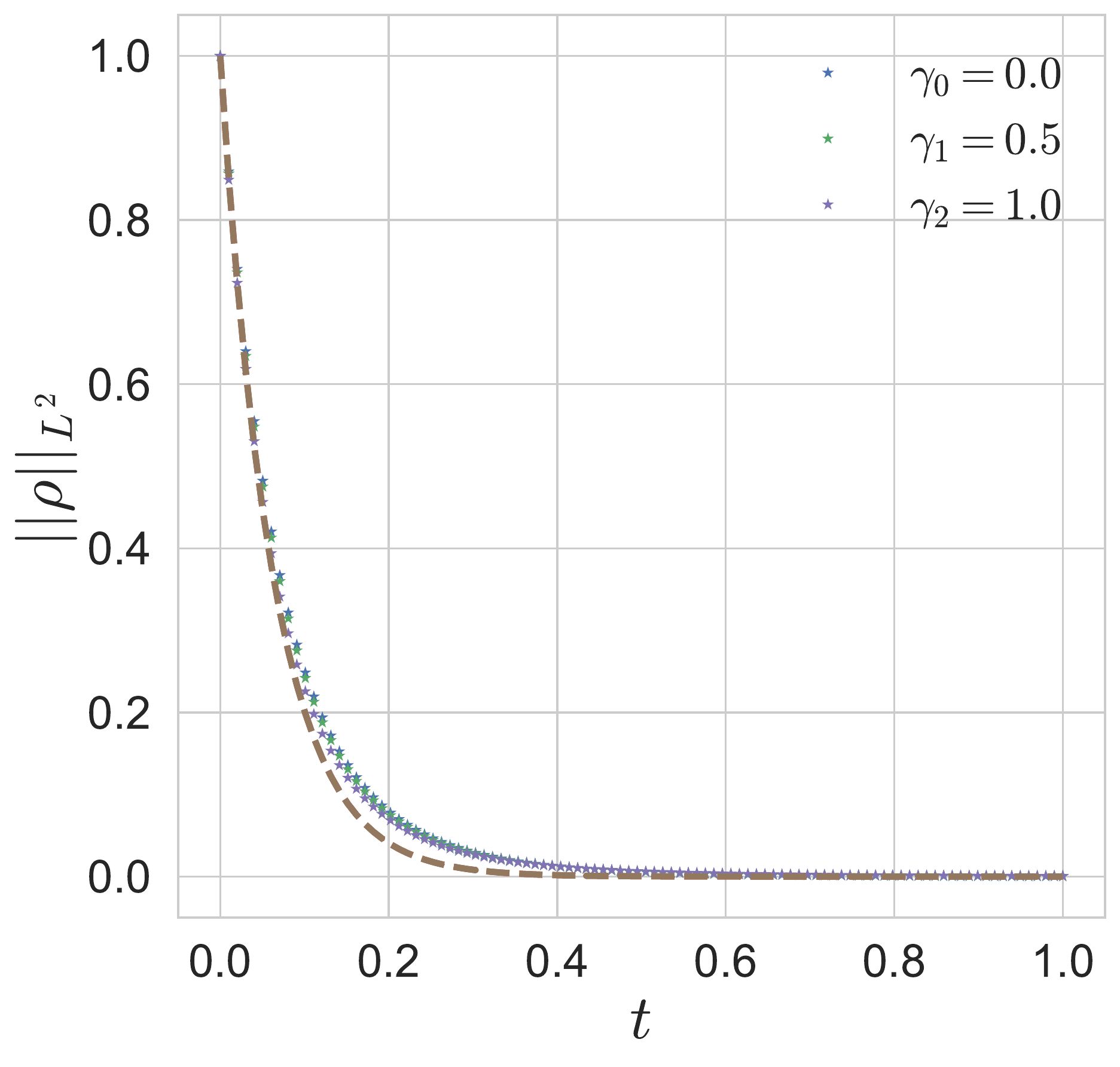}
\end{center}
\caption{Rate of Decoherence. Decay of $\tr[\hat \rho^*(t)\hat\rho(t)]/\tr[\rho_0^*\rho_0]$ for an oscillatory initial state $\rho_0$ in the form \eqref{eq:Gaussian-0} with $p_0=q_0=0$ and $\delta p=2, \delta q=4$ for $\sigma=1$, $\hbar=1$ and different values of $\gamma$: $\gamma=0$, $\gamma=0.5$, and $\gamma=1$. 
We compare this with the short time prediction \eqref{small_time_trace_rho} (dashed line).  }\label{fig:S_l2norm_evolution_p0_0}
\end{figure}

In Figure \ref{fig:S_l2norm_evolution_p0_0} we have computed the value of 
$\frac{\tr[\hat\rho^*(t)\hat\rho(t)]}{\tr[\hat \rho_0^*\hat\rho_0]}$
for an oscillatory initial $\rho_0$ with $\delta p\neq 0$ and $\delta q$ and find excellent agreement with \eqref{small_time_trace_rho} over a range of values of $\gamma$. In particular the onset of decoherence is captured very accurately, the small deviations for slightly larger times  are most likely due to the short time approximations we made in the derivation of \eqref{small_time_trace_rho}.

\subsection{Propagation of general states}

We have focused so far on the case that the initial state is a Gaussian, or a sum of Gaussians. But due to the fact that the equation for $w(t)$, \eqref{eq:w_Lindblad}, is quadratic in $p$ and derivatives in $p$ we can write down an explicit formula for the propagator $K(t,p,p')$ of \eqref{eq:w_Lindblad} which is defined by the relation 
\begin{equation}
w(t,p)=\int K(t,p,p')w_0(p')\, \ud p'\,\, ,
\end{equation}
where $w_0(p)$ is the initial value at $t=0$ of the partial Fourier transform of the Wigner function of the state. We will follow the ideas in \cite{Hor95} and make an Ansatz for $K$ as the Weyl-quantisation of a function 
$A(t, p,\xi)$, i.e.,
\begin{equation}\label{eq:heat-kernel}
K(t, p,p')=\frac{1}{2\pi\hbar}\int \ue^{\frac{\ui}{\hbar}(p-p')\xi}A\bigg(t, \frac{p+p'}{2}, \xi\bigg)\ud\xi\,\, 
\end{equation}
and then the formalism of semiclassical analysis gives for $A$ the equation
\begin{equation}
\ui\hbar \pa_t A=H\# A\,\, ,\quad\text{where}\quad H\# A=\sum_{n=0}^{\infty}\frac{\hbar^n}{2^nn!}H(\overleftarrow{\nabla} \cdot\Omega \overrightarrow{\nabla})^nA\,\, 
\end{equation}
where $H$ is given by \eqref{eq:mod-Ham}. Since $H$ is quadratic the sum for $H\#A$ will terminate after $n=2$. We will also use that the product is invariant under linear canonical transformations (see, for instance \cite{Zwo12}), i.e., if $\varphi:\R^2\to\R^2$ is a linear canonical transformation and $A_{\varphi}:=A\circ \varphi$, $H_{\varphi}=H\circ \varphi$, then 
$(H\#A)\circ \varphi=H_{\varphi}\#A_{\varphi}$. We can use this observation to simplify $H$ by letting $\eta\neq 0$ and choosing $\varphi(\xi, p)=(\lambda \xi, p/\lambda +p_0)$ with 
$p_0=-\ui/(\gamma^2\eta)$ and $\lambda^2=\gamma\eta/\sigma$. Then, a short calculation gives that for $H$ in \eqref{eq:mod-Ham}
\begin{equation}
H_{\varphi}(\xi,p)=-\ui\frac{\omega}{2}(\xi^2+p^2)+\frac{\ui}{2\gamma^2}\,\, ,
\end{equation}
where $\omega=\sigma\gamma\eta$, hence for $\eta\neq 0$ the Hamiltonian is canonically equivalent to a complex harmonic oscillator. The solution to $\ui\hbar\pa_t A_{\varphi}=H_{\varphi}\#A_{\varphi}$ with initial condition $A_{\varphi}(t=0)=1$ is known (see, for instance \cite{Hor95}) and given by 
\begin{equation}
A_{\varphi}(t,\xi,p)=\ue^{-\frac{1}{2\gamma^2\hbar}[t-\frac{2\tanh(t\hbar\omega/2)}{\omega}]}\ue^{-\frac{\ui}{\hbar}\frac{\tanh(t\hbar\omega/2)}{\omega}H_{\varphi}(\xi,p)}\,\, .
\end{equation}
And by composing with the inverse of $\varphi$ we obtain
\begin{equation}\label{eq:symbol-heat-kernel}
A(t,\xi,p)=\ue^{-\frac{1}{2\gamma^2\hbar}[t-\frac{2\tanh(t\hbar\omega/2)}{\omega}]}\ue^{-\frac{\ui}{\hbar}\frac{\tanh(t\hbar\omega/2)}{\omega}H(\xi,p)}\,\, .
\end{equation}
where $H$ is given by \eqref{eq:mod-Ham} and $\omega=\sigma \gamma\eta$. Although in the derivation we made the assumption $\eta\neq 0$, the final result can be extended to 
$\eta=0$ by continuity. This general representation of the non-Gaussian channel is another of the main results of this paper. 

Notice that \eqref{eq:symbol-heat-kernel} is quadratic in $\xi$ so we could perform the $\xi$-integral in \eqref{eq:heat-kernel} and obtain an explicit formula for the propagator $K(t,p,p')$ which would be similar to the Mehler formula and would give us the analog of the results for the heat-kernel of the Grushin operator, \cite{Chang15}, but we follow the philosophy of H{\"o}rmander in \cite{Hor95} that the Weyl symbol gives a more natural representation of the propagator.  



\section{Approximations}

In the previous section we derived an explicit representation of the  characteristic function of a non-Gaussian Quantum Channel acting on a Gaussian state. One application of this formula is that it allows us to test the accuracy of approximations. In this section we will look at two natural approximations, the semiclassical approximation which can be applied to general Lindblad equations, and the small $\gamma$ expansion, where we look at the first order correction of \eqref{eq:Lindblad-double-commutator} to the Gaussian Channel defined by $\gamma=0$.

\subsection{Semiclassical Approximation}\label{sec:semi-approx}

In \cite{GraeLongPlasSchu18}, see as well \cite{BroOzo10}, the authors developed a semiclassical approximation for the Lindblad evolution if the Wigner function of the initial state is a (linear combination of) localised Gaussians of the form
\begin{equation}\label{eq:gen-initial-cond}
\rho(x)=\frac{1}{(\pi\hbar)^n\sqrt{\det G }}\ue^{-\frac{1}{\hbar}(x-X)\cdot G^{-1}(x-X)+\frac{\ui}{\hbar}x\cdot Y}
\end{equation}
where $X,Y\in\R^n$ are parameters, and $G$ is a positive symmetric matrix which satisfies the Robertson-Schr{\"o}dinger uncertainty relation $G+\ui\Omega\geq 0$. Notice that for $n=1$ this initial state matches \eqref{eq:Gaussian-0} if we choose in \eqref{eq:gen-initial-cond} $X=(p_0,q_0)$, $Y=(-\delta q, \delta p)$ and $G=\begin{pmatrix} g & 0\\ 0 & 1/g\end{pmatrix}$.  The main idea in \cite{GraeLongPlasSchu18} is to interpret the evolution equation for the Wigner function as a Schr{\"o}dinger equation with non-Hermitian Hamiltonian given by 
\begin{equation}\label{eq:K}
K(x,y)=-y\cdot Fx-\frac{\ui}{2}\big[(l\cdot y)^2+(y\cdot Fx)^2\big]
\end{equation}
which can then be solved using the methods developed for non-Hermitian propagation in \cite{GraeSchu11,GraeSchu12}. Here $F=\Omega Q$ is the Hamiltonian map of $H(x)=\frac{1}{2} x\cdot Qx$ and $l$ is related to $L(x)=l\cdot \Omega x$, and $y$ is the momentum variable dual to $x$. For the case $Q=\begin{pmatrix}1 & 0\\0 & 0\end{pmatrix}$ and $l=(0,1)$ the non-Hermitian Schr{\"o}dinger equation $\ui\hbar \pa_t \rho=K(x,\hat x)\rho$ gives \eqref{eq:Lindblad_model-case}.

Let us first consider the case that $Y=0$, then in the leading order semiclassical approximation the state stays in the form \eqref{eq:gen-initial-cond} and the parameters $X$ and $G$ satisfy the equations
\begin{align}
\dot X&=\Omega \nabla H(X)\label{eq:X-dot}\\
\dot G&=\Omega H''(X) G-GH''(X) \Omega+2\Omega^TD(X)\Omega\label{eq:G-dot}
\end{align}
where $H''(X)$ is the Hessian of $H$ at $x=X$ and $D(X)=\sum_{k}\nabla L_k(X)\nabla L_k(X)^T$, see \cite{GraeLongPlasSchu18}. Here we have assumed in addition that all the Lindblad operators $L_k$ are Hermitian, and in \cite{GraeLongPlasSchu18} the equation is given for $G^{-1}$ instead of $G$. In our situation the Hamiltonian is quadratic, $H=\frac{1}{2} x\cdot Qx$, and then the first equation, \eqref{eq:X-dot}, is solved by 
\begin{equation}
X(t)=R_tX_0\,\, ,\quad\text{where}\quad R_t=\ue^{t\Omega Q}\,\, 
\end{equation}
and inserting  an ansatz $G(t)=R_t \Lambda(t)R_t^T$ into the second equation, \eqref{eq:G-dot}, yields
\begin{equation}
\dot \Lambda=2R_{-t}\Omega^TD(X(t))\Omega R_{-t}^T=2\sum_k R_{-t}\Omega\nabla L_k(X)[R_{-t}\Omega\nabla L_k(X)]^T\,\, .
\end{equation}
Here we have used that $D$ is a sum of all the Lindblad operators, and so we can integrate each term separately. If $L_k(x)=l_k \cdot \Omega x$ is linear, then $R_{-t}\Omega\nabla L_k(X)=R_{-t} l_k$ is independent of $x$ and the contribution will be identical to the Gaussian Channel case. If $L_k=\alpha H$, with $H=\frac{1}{2}x\cdot Qx$ quadratic, then $R_{-t}\Omega Q=\Omega QR_{-t}$, as $R_t=\ue^{t\Omega Q}$, and hence 
\begin{equation}
R_{-t}\Omega \nabla L(X(t))=\alpha R_{-t}\Omega QR_t X_0=\Omega QX_0
\end{equation}
is time independent and can be easily integrated. So for the case that $L_1=\sigma l\cdot \Omega x$ and $L_2(x)=\gamma H(x)=\frac{\gamma}{2}x\cdot Qx$ we get 
\begin{equation}
\Lambda_t=G_0+2\sigma^2\int_0^t R_{-s}l[R_{-s}l]^T\, \ud s+t\gamma^2 FX_0[FX_0]^2\,\,\ud s\,\,  , 
\end{equation}
where again, $F=\Omega Q$. For our case we get with  $l=(0,1)^T$ and $F=\begin{pmatrix} 0 & 1 \\ 0 & 0\end{pmatrix}$ 
\begin{equation}\label{eq:Lambda_t}
\Lambda_t=G_0+2\sigma^2\begin{pmatrix} \frac{1}{3}t^3 & -\frac{1}{2}t^2 \\-\frac{1}{2} t^2 & t\end{pmatrix}+2t \gamma^2 p_0^2 \begin{pmatrix} 1 & 0\\0 & 0\end{pmatrix} \,\, ,
\end{equation}
where $X_0=(q_0,p_0)^T$ and $R_t=\begin{pmatrix} 1 & t\\ 0 & 1\end{pmatrix}$, and for $G_0=\begin{pmatrix}1/g & 0 \\ 0 & g\end{pmatrix}$ this  finally leads to
\begin{equation}\label{eq:G-sc}
    G_t=\begin{pmatrix} \frac{1}{g}+2t\gamma^2p_0^2+gt^2+\frac{2}{3}\sigma^2 t^3 &gt+\sigma^2 t^2\\ gt +\sigma^2 t^2 &g+2\sigma^2 t\end{pmatrix}\,\, .
\end{equation}
We can compare this semiclassical approximation with the exact variances $\Gamma$ we have computed in \eqref{eq:full-variance-q},\eqref{eq:full-variance-pq} and \eqref{eq:full-variance-p}, and see that
\begin{equation}\label{eq:Gamma-vs-sc}
    \Gamma=\frac{\hbar}{2}G_t+\frac{\hbar^2}{2}\begin{pmatrix} \gamma^2 (gt+\sigma^2 t^2) & 0\\ 0 & 0\end{pmatrix}\,\, .
\end{equation}
So the semiclassical approximation is correct for the order $\hbar$ terms, but is not catching the next order $\hbar^2$ term, which is proportional to $\gamma^2$. 

In Figure \ref{fig:rho_gaussian_evolution_p0_0dot5} we compare the semiclassical approximation to the exact result for a state which has a non-zero initial momentum. We see that the semiclassical approximation accurately reproduces the motion of the centre, but shape of the state becomes less Gaussian as time evolves.

We will now consider the semiclassical approximation for the case $Y\neq 0$, which we will write as 
\begin{equation}\label{eq:rho-X-Y}
\rho(x,t)=\frac{c\ue^{\frac{\ui}{\hbar}d}}{(\pi\hbar)^n\sqrt{\det  G }}\ue^{-\frac{1}{\hbar}(x-X)\cdot G^{-1}(x-X)+\frac{\ui}{\hbar}x\cdot Y}
\end{equation}
where $c,d\in \C$, $Z=(X,Y)\in\R^{4n}$ are time dependent parameters and $G$ is complex symmetric with $\Re G>0$. The results in \cite{GraeLongPlasSchu18}, equations (68), (69) and (70), can then be rewritten as 
\begin{align}
\dot Z&=\Omega_2 \nabla \Re K(Z)+\mathcal{G}^{-1} \nabla\Im K(Z) \label{eq:dot-Z}\\
\dot G&=2\ui K_{yy}+K_{yx} G+GK_{xy}-\frac{\ui}{2}G K_{xx}G\label{eq:dot-G}\\
\dot d&=-K(Z)+\dot Y X\\
\dot c&=\frac{1}{4}\tr[2\ui K_{yy}G^{-1}+\frac{\ui}{2} K_{xx}G] 
\end{align}
where by equation (76) in \cite{GraeLongPlasSchu18}
\begin{equation}
\mathcal{G}^{-1}=\begin{pmatrix} \Re G+\Im G(\Re G)^{-1}\Im G & -\Im G(\Re G)^{-1}\\ -(\Re G)^{-1}\Im G & (\Re G)^{-1}\end{pmatrix}
\,\, .
\end{equation}
This is a more complex set of evolution equations for the parameters than in the case $Y=0$, due to the fact that now the equation for $Z$ contains $G$, and hence cannot be solved independently from the equation for $G$, which in turn depends on $Z$ because the matrices of second derivatives of $K$ are evaluated at $Z$. This is a characteristic property of non-Hermitian evolution, as discussed in \cite{GraeSchu11,GraeSchu12} and \cite{BurLupUribe13}.

\begin{figure}[ht]
\centering
\subfigure[$\gamma=1,t=0$]{\includegraphics[width=0.3\linewidth]{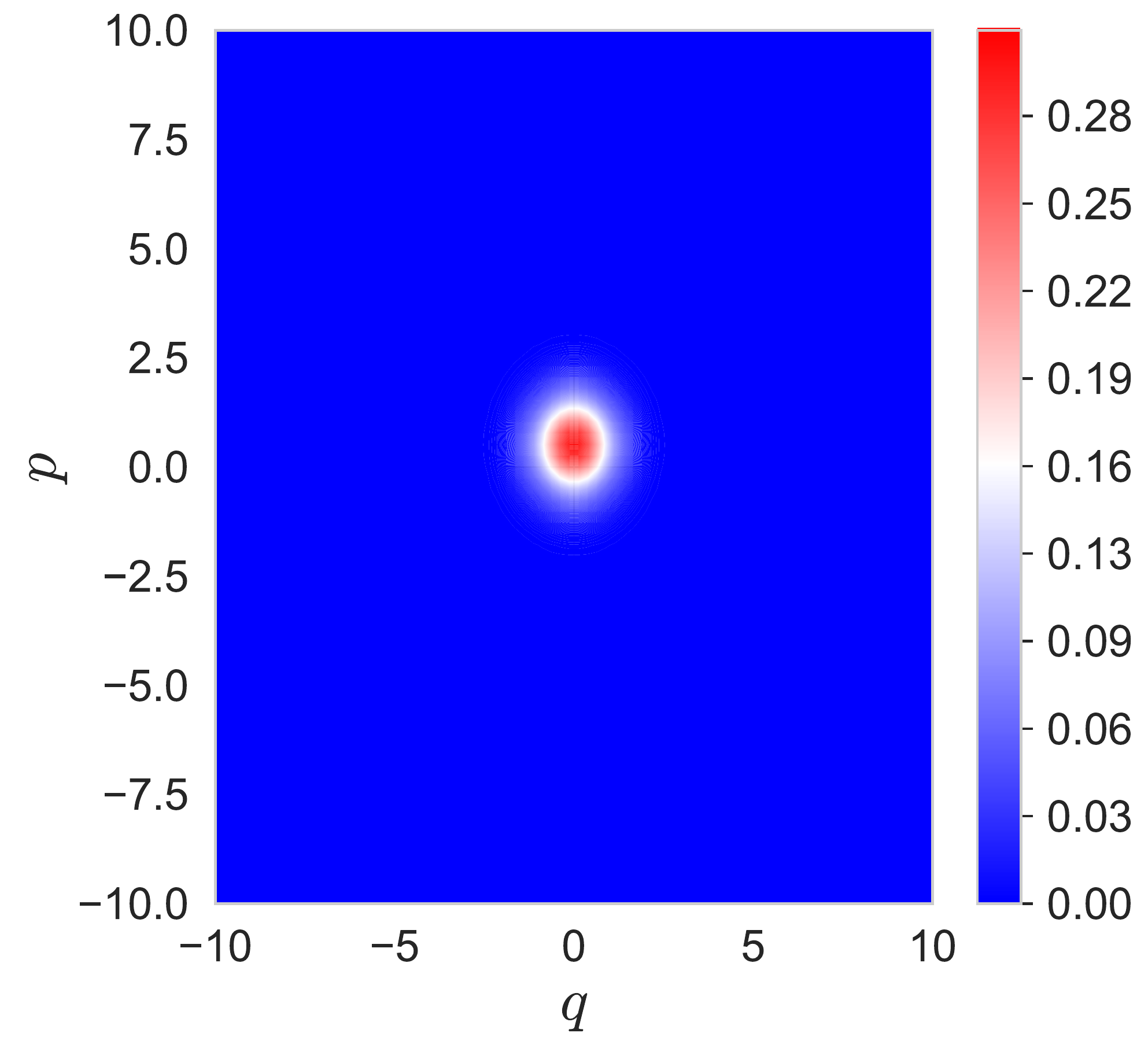}}
\subfigure[$\gamma=1,t=0.5$]{\includegraphics[width=0.3\linewidth]{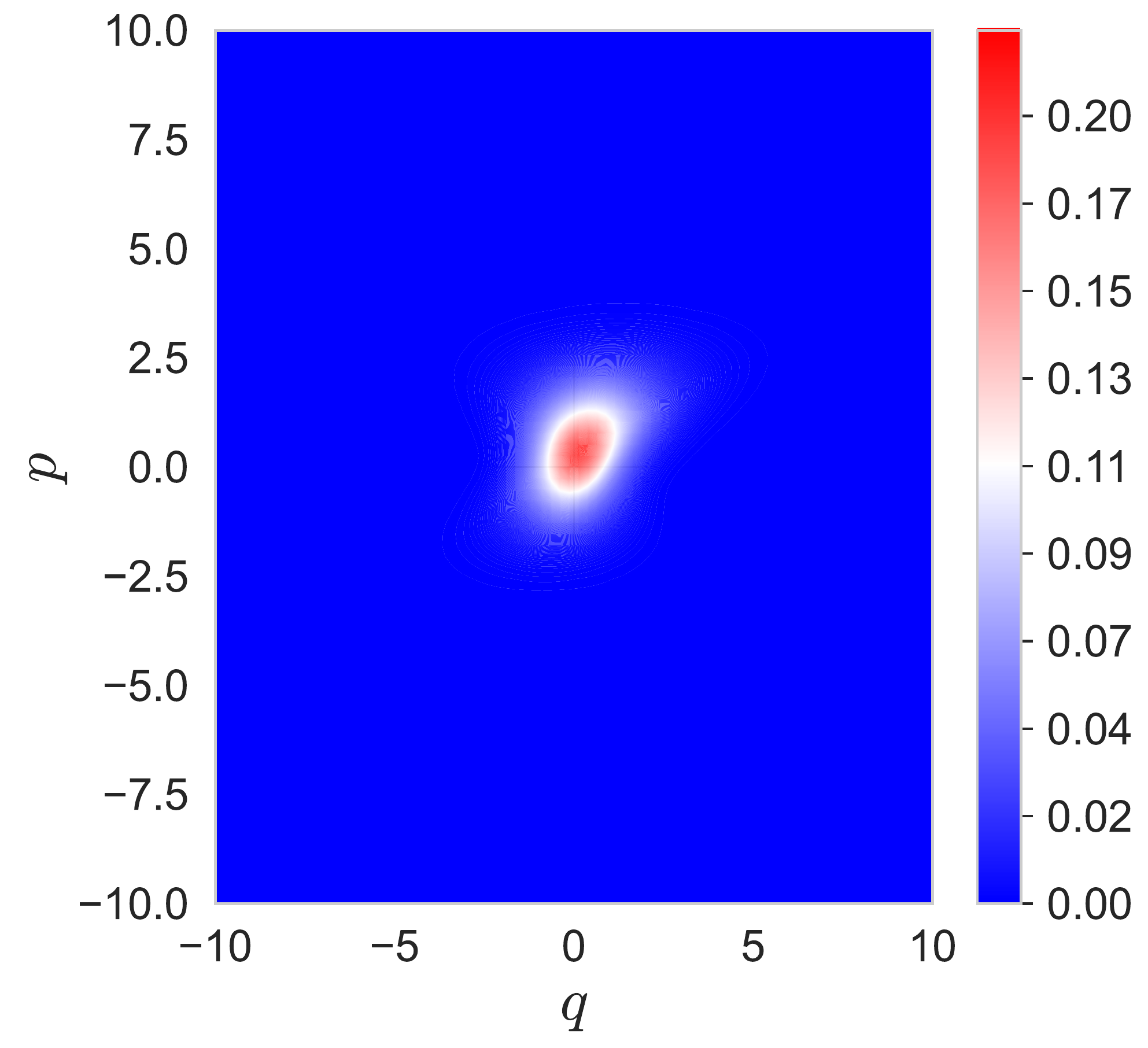}}
\subfigure[$\gamma=1,t=1$]{\includegraphics[width=0.3\linewidth]{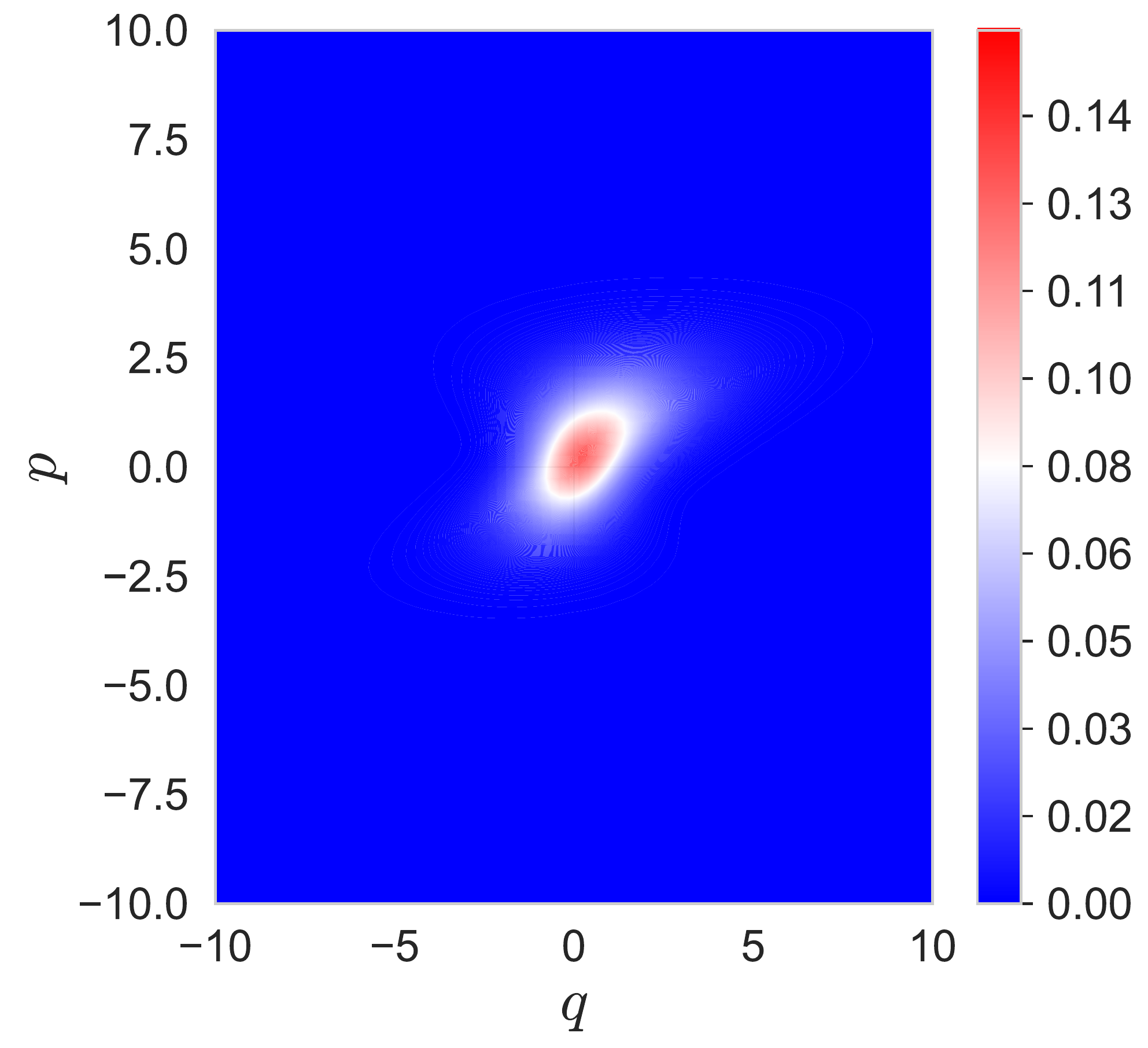}}
\subfigure[$\gamma=1,t=0$, Gaussian]{\includegraphics[width=0.3\linewidth]{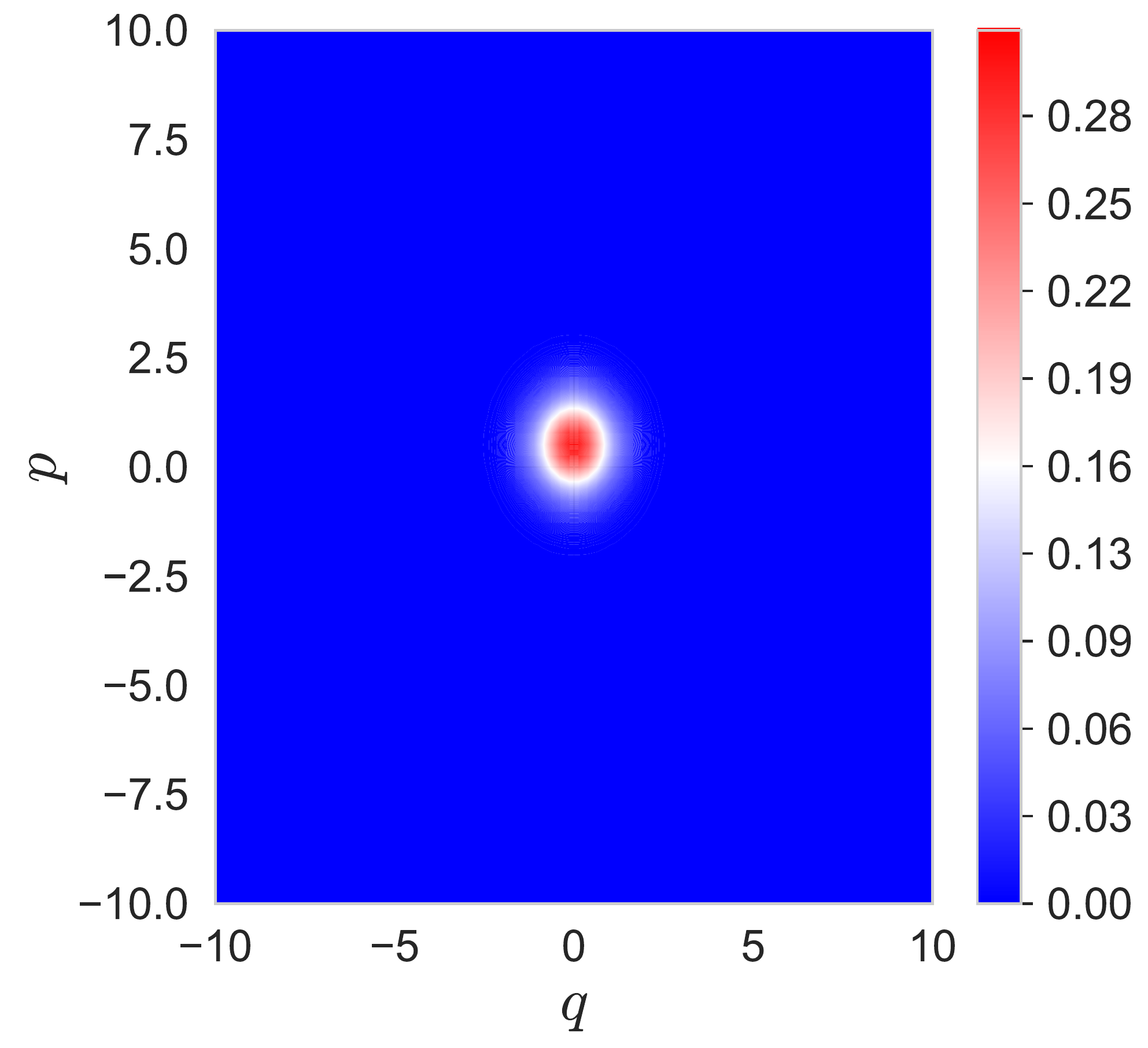}}
\subfigure[$\gamma=1,t=0.5$, Gaussian]{\includegraphics[width=0.3\linewidth]{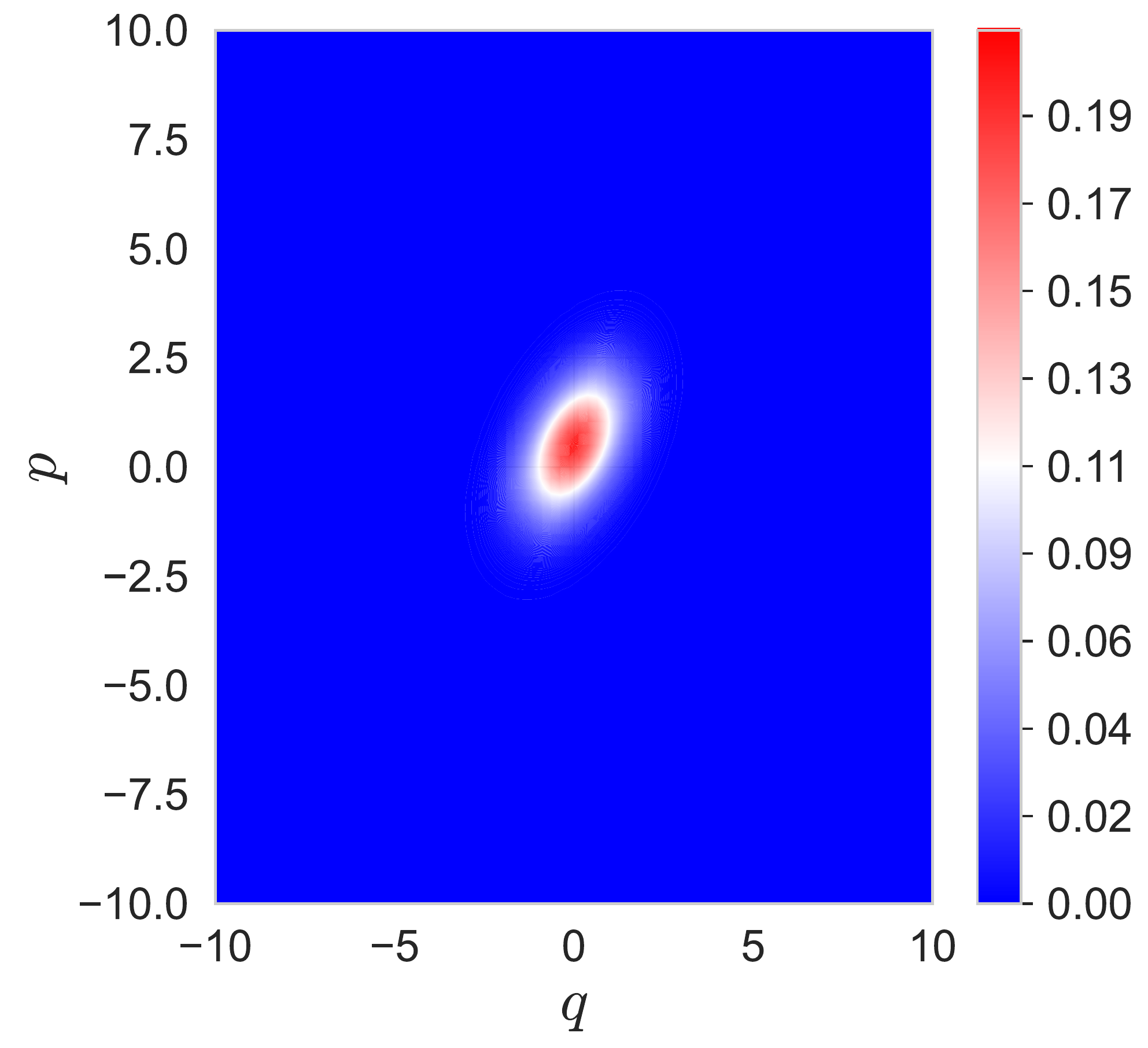}}
\subfigure[$\gamma=1,t=1$, Gaussian]{\includegraphics[width=0.3\linewidth]{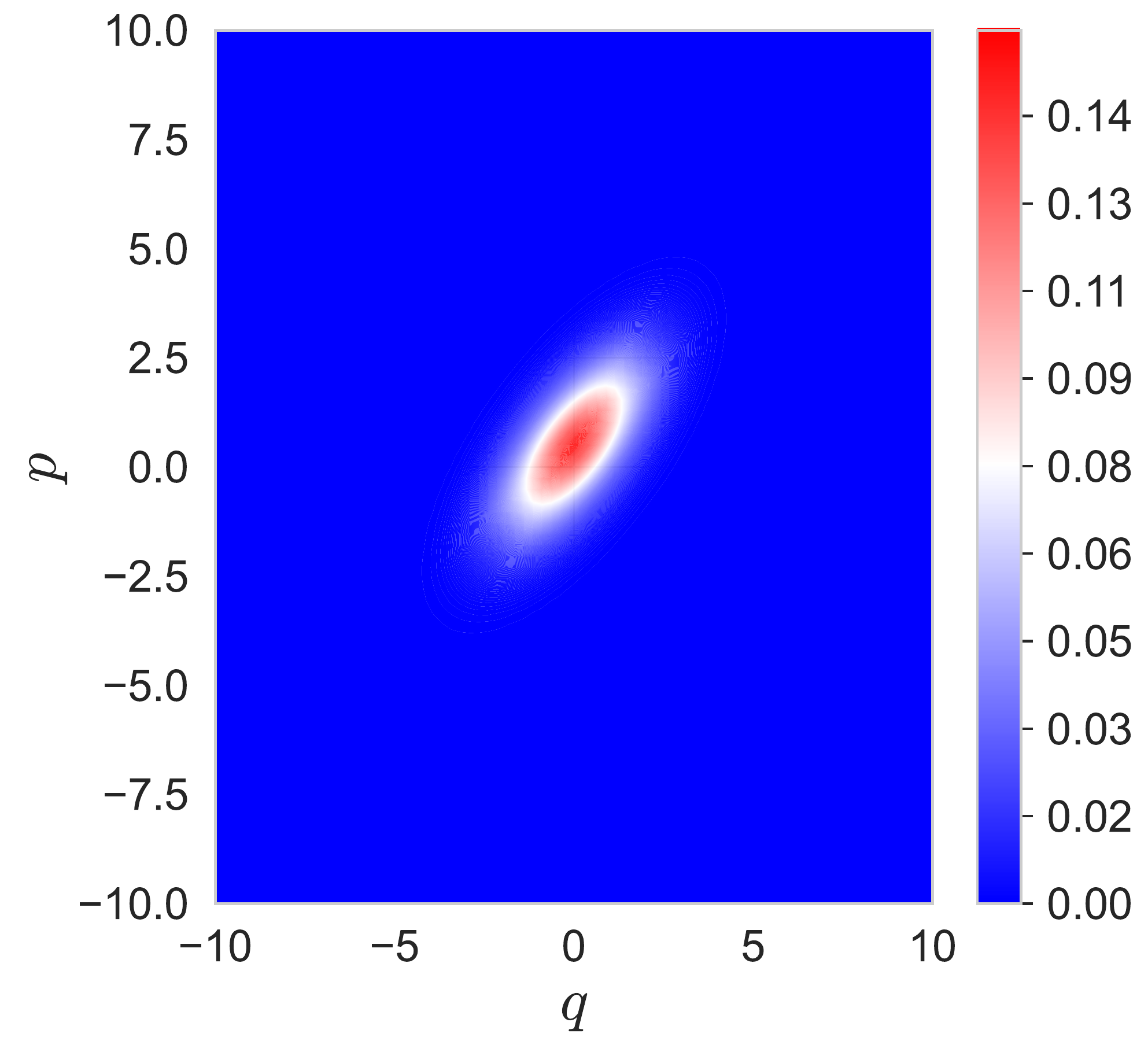}}
\caption{Comparison of evolution of $\rho$ and its Gaussian approximation based on \eqref{eq:G-sc}  for an initial state \eqref{eq:Gaussian-0} with $q_0=0$, $p_0=0.5$, $g=1$ and $\hbar=1$ for $\sigma=1$ and $\gamma=1$. }\label{fig:rho_gaussian_evolution_p0_0dot5}
\end{figure}

But for $Y\neq 0$ we expect decoherence to set in and suppress $\rho(x,t)$ rapidly, as we have seen in \eqref{small_time_trace_rho}, and we now want to study how we can reproduce \eqref{small_time_trace_rho} from our semiclassical equations. 
We have for \eqref{eq:rho-X-Y} that 
\begin{equation}
\tr[\hat\rho^*\hat\rho]=\frac{\abs{c}^2}{(\det G)^{1/2}} \, \ue^{-\frac{2}{\hbar} \Im d}\,\, ,
\end{equation}
so we get an exponential decay if $\Im d(t)>0$. As $\Im d(t)=\int_0^t \Im K(X(s),Y(s))\, \ud s$ and $\Im K(X,Y)\geq 0$ we have that $d(t)$ is non-decreasing for increasing $t$ and in particular we have that if $\Im K(X_0,Y_0)\neq 0$ then 
\begin{equation}
\Im d(t)=t \Im K(X_0,Y_0)+O(t^2)\,\, .
\end{equation}
If $\Im K(X_0,Y_0)=0$ then $\nabla \Im K(X_0,Y_0)=0$, as $\Im K(Z)\geq 0$, and in this case \eqref{eq:dot-Z} gives $\dot X(t=0)=FX_0$ and $\dot Y(t=0)=-F^T Y_0$, which leads to 
\begin{equation}
\Im K(X,Y)=\frac{1}{2}\big[t^2(Fl\cdot Y_0)^2+4t^2(Y_0F^2X_0)^2\big]+O(t^3)\,\, .
\end{equation}
In our case $F^2=0$ and $Fl\cdot Y_0=\gamma e_1\cdot Y_0$, where $e_1=(1,0)$ and $e_2=(0,1)$, so we have 
\begin{equation}
\Im d(t)=\begin{cases} t\Im K(X_0,Y_0) +O(t^2) & \Im K(X_0,Y_0)\neq 0\,\, ,\\ \frac{1}{6}t^3 \gamma^2(e_1\cdot Y_0)^2+O(t^4) & \Im K(X_0,Y_0)=0\,\, .\end{cases}
\end{equation}
From $\Im K(X,Y)=\frac{1}{2}[\gamma^2(e_2\cdot Y)^2+\alpha^2p^2(e_1\cdot Y)^2]$ we see that this exactly reproduces the results in \eqref{eq:D-short-time} and \eqref{small_time_trace_rho} we obtained from the exact evolution. So the semiclassical approach predicts the onset of decoherence correctly.

\subsection{Perturbation theory for small dephasing}\label{sec:dephasing}

If the dephasing  parameter $\gamma$ is small, we expect that the channel is close to Gaussian. One possibility to explore the regime of small $\gamma$ is to use the explicit representations we obtained in Section \ref{sec:Gaussian-prop}. But as we saw, the limit for small $\gamma$ is quite complicated. Therefore, in this section we will use time dependent perturbation theory to derive an expression for small $\gamma$. As a starting point we will rewrite the Lindblad equation \eqref{eq:Lindblad-V} as 
\begin{equation}
\pa_t\rho=\cL_1\rho+\gamma^2\cL_2\rho\,\, ,\quad\text{where}\quad \cL_1=-V_0+\frac{\hbar\sigma^2}{2}V_1^2\,\, ,\quad \cL_2=\frac{\hbar}{2}V_0^2
\end{equation}
and the vector fields are defined to be
\begin{equation}\label{eq:V-com}
V_0=p\pa_q\,\, ,\quad V_1=\pa_p\,\, ,\quad V_2=\pa_q\,\, ,\quad\text{with}\quad [V_0,V_1]=-V_2\,\, ,\quad [V_0,V_2]=[V_1,V_2]=0\,\, .
\end{equation}
Now standard time dependent perturbation theory (see for instance \cite{Kato80}) gives 
\begin{equation}\label{eq:exp-gamma}
\rho(t)=\ue^{t\cL_1}\rho_0+\gamma^2\cR_1(t)\ue^{t\cL_1}\rho_0+O(\gamma^4) \,\, \quad\text{with}\quad \cR_1=\int_0^t\ue^{s\cL_1}\cL_2\ue^{-s\cL_1}\,\, \ud s
\end{equation}
and we can expand the operator $\cR_1(t)$ using $\ue^{s\cL_1}\cL_2\ue^{-s\cL_1}=\sum_{k=0}^{\infty}\frac{s^k}{k!}\ad_{\cL_1}^k \cL_2$, where 
\begin{equation}
\ad_{\cL_1}\cL_2:=[\cL_1,\cL_2]\,\, ,
\end{equation}
and $\ad_{\cL_1}^0\cL_2=\cL_2$, which gives
\begin{equation}\label{eq:exp-R}
\cR_1(t)=\sum_{k=0}^{\infty}\frac{t^{k+1}}{(k+1)!} \ad_{\cL_1}^k \cL_2\,\, .
\end{equation}
In general this representation as an infinite sum is not very useful, but in our situation the sum is actually finite and terminates at $k=4$. We find using \eqref{eq:V-com} that
\begin{equation}
\ad_{\cL_1}\cL_2=\frac{\hbar^2\sigma^2}{2}(V_0V_1+V_1V_0)V_2\,\, ,\quad  \ad_{\cL_1}^2\cL_2=\hbar^2\sigma^2(V_0+\hbar\sigma^2V_1^2)V_2^2
\end{equation}
and
\begin{equation}
\ad_{\cL_1}^3\cL_2=3\hbar^3\sigma^4V_1V_2^3\,\, \quad  \ad_{\cL_1}^4\cL_2=-3\hbar^3\sigma^4V_2^4\,\, , \quad \ad_{\cL_1}^5\cL_2=0\,\, ,
\end{equation}
and then all higher terms vanish. So $\cR_1(t)$ is a differential operator of order 4 in $x=(q,p)$ which acts on the Wigner function $\big(\ue^{t\cL_1}\rho_0\big)(x)$ and hence we can compute $\cR_1(t) \big(\ue^{t\cL_1}\rho_0\big)(x)$ if, for instance, $\rho_0(x)$ is Gaussian. We will not compute the expression explicitly here, but we will use this result in the next section where we consider the entropy of the time evolved states. 



\section{Entropy}

In this section we will look at the entropy of our states and how it evolves in time. Recall that the von Neumann entropy of a state $\hat\rho$ is defined as 
\begin{equation}
S(\hat\rho):=-\tr\big[ \hat\rho\ln\hat \rho\big]\,\, ,
\end{equation}
and is an extension of the Shannon entropy from information theory to quantum mechanics. The entropy is $0$ if and only if the state $\hat \rho$ is pure, so the entropy can as well be viewed as quantifying how much the state $\hat \rho$ is away from being pure. 

The entropy of a Gaussian state can be expressed explicitly in terms of the symplectic eigenvalues of the covariance matrix $G$. If $G$ is a real symmetric positive $2n\times 2n$ matrix, then the eigenvalues of $\ui\Omega^TG$ come in pairs $\pm z_j$, $j=1,2,\cdots, n$, $z_j>0$, which are called the symplectic eigenvalues of $G$. By the uncertainty relation we have $z_j\geq 1$.  The entropy of the state is then given by \cite{HolWer01}
\begin{equation}\label{eq:S-quadratic}
S(\hat \rho)=\sum_{j=1}^nf(z_j) \,\, ,\quad \text{where}\quad  f(z)=\frac{1}{2}(z+1)\ln\bigg(\frac{1}{2}(z+1)\bigg)-\frac{1}{2}(z-1)\ln\bigg(\frac{1}{2}(z-1)\bigg)   .
\end{equation}
Notice that our convention for the normalisation of $G$ is different from \cite{HolWer01}, which leads to the extra factors of $2$ in $f$. In the special case that $n=1$ we have $z_1=\sqrt{\det G}$ and so $S(\hat \rho)=f(\sqrt{\det G})$. We can therefore use the Gaussian approximation for $\hat \rho$ from Section \ref{sec:semi-approx}, which gives 
$z_1=\sqrt{\det G}=\sqrt{\det \Lambda_t}$ with $\Lambda_t$ given by 
\eqref{eq:Lambda_t}.

We want to compare this with the approximation for small dephasing we developed in Section \ref{sec:dephasing}, to that end we will make use of the following result.

\begin{lem}\label{lem:S-FH} Assume $\hat \rho(\veps)$, $\veps \geq 0$,  is differentiable family of density operators with $\tr \hat \rho(\veps)=1$, then 
\begin{equation}
S(\hat \rho(\veps))=S(\hat\rho(0))-\veps \tr\big[\hat \rho'(0)\log\hat \rho(0)\big]+O(\veps^2)\,\, ,
\end{equation}
where $\hat \rho'(\veps)=\frac{\ud \rho(\veps)}{\ud \veps}$.
\end{lem}

\begin{proof}
Let $\lambda_n(\veps)$ and $|n,\veps\ra$ be the eigenvalues and eigenfunctions of $\hat \rho(\veps)$, then by the Feynman Hellman theorem we have 
$\lambda_n'(\veps)=\la n,\veps|\hat\rho'(\veps)|n,\veps\ra$, and so we get 
\begin{align}\label{deriv_S_rho}
\frac{\ud S(\hat\rho(\veps))}{\ud \veps}
&=-\sum_n\frac{\ud \big(\lambda_n(\veps)\ln\lambda_n(\veps)\big)}{\ud \veps}=-\sum_n \la n,\veps|\hat\rho'(\veps)|n,\veps\ra\ln \lambda_n(\veps)-\sum_n \lambda_n'(\veps) \nonumber \\
&=-\tr \big[\hat \rho'(0)\log\hat \rho(0)\big]
\end{align}
where we used that $\sum_n \lambda_n'(\veps)=0$ since $\tr \hat\rho(\veps)=1$. The result is then obtained by substituting the expression \eqref{deriv_S_rho} into the Taylor expansion of $S(\hat{\rho}(\veps))$ for small $\veps$.
\end{proof}

We will use this result with $\veps=\gamma^2$ and combine it with \eqref{eq:exp-gamma} to give an expression for $\hat \rho'(0)$. We will also use some results from semiclassical analysis which we recalled in Appendix \ref{app:semi}. In order to explain the main ideas let us first look at the first term, $S(\hat \rho(0))$, where $\hat \rho(0)$ is a Gaussian state \eqref{eq:Gaussian-state-diag-G} with covariance matrix $G$. The Weyl symbol $\rho$ of $\hat \rho(0)$ is given by $2(\det G)^{-1}\ue^{-\frac{1}{\hbar} x\cdot G^{-1}x}$ and so we can use \eqref{eq:prod-tr} to obtain 
\begin{equation}
S(\hat \rho(0))=-\tr\big[\hat \rho(0)\ln \hat\rho(0)\big]=-\frac{1}{2\pi\hbar}\int 2(\det G)^{-1}\ue^{-\frac{1}{\hbar} x\cdot G^{-1}x} B(x)\,\, \ud x
\end{equation}
where $B(x)$ denotes the Weyl symbol of $\ln \hat \rho(0)$, i.e., $\hat B= \ln \hat \rho(0)$. We show in Appendix \ref{app:semi} that $B(x)=-\frac{1}{2}x\cdot Qx-\ln Z$ with 
\begin{equation}\label{eq:Q-G}
Q=2\frac{z}{\hbar}\coth^{-1}(z)\, G^{-1}\,\, ,\quad \text{where}\quad z=\sqrt{\det G}\,\, ,
\end{equation}
and 
\begin{equation}
Z=\frac{1}{2}\sqrt{z^2-1}\,\, .
\end{equation}
Using \eqref{eq:int-tr-GQ} and the normalisation of $\hat \rho(0)$ we then find
\begin{equation}
S(\hat \rho(0))=\frac{\hbar}{2}\tr[GQ]+\ln Z(1)=z\coth^{-1}(z)+\ln\bigg(\frac{1}{2}\sqrt{z^2-1}\bigg)
\end{equation}
and after a bit of algebra this expression can be transformed into \eqref{eq:S-quadratic} for $n=1$. 

To compute the second term in the expansion we follow the same strategy and use that the Weyl symbol of $\hat \rho'(0)$ is given by $\cR_1\rho$ where $\cR_1$ is a differential operator in $x$ given by \eqref{eq:exp-gamma}, then we have 
\begin{equation}
\tr [\hat\rho'(0)\ln \hat \rho(0)]=\frac{1}{2\pi\hbar}\int \cR_1\rho(x)B(x)\, \ud x=\frac{1}{2\pi\hbar}\int \rho(x)\cR_1^{T}B(x)\, \ud x\,\, ,
\end{equation}
where $\cR_1^{T}$ is the adjoint of $\cR_1$ obtained by partial integration. The fact that $B(x)$ is a polynomial of order $2$ implies that terms in $\cR^T$ containing derivatives of order $3$ and higher will not contribute to $\cR_1^{T}B(x)$ and in fact only the first two terms in the expansion \eqref{eq:exp-R} contribute, which gives 
\begin{equation}
\cR_1^{T}B(x)=t\cL_2^{T}B(x)+\frac{1}{2}t^2[\cL_1,\cL_2]^{T}B(x)=t\frac{\hbar}{2}V_0^2B(x)-\frac{1}{2}t^2\frac{\hbar^2\sigma^2}{2}V_2^2B(x)\,\, .
\end{equation}
With $V_0^2B(x)=-p^2Q_{22}$ and a $V_2^2B(x)=-Q_{22}$, where $Q_{22}$ is the $q,q$ matrix element of $Q$, we find 
\begin{equation}
\tr[\hat\rho'(0)\ln \hat \rho(0)]=-\frac{1}{2\pi\hbar}\int \rho(x)\frac{1}{2}x\cdot E_p x\, \ud x\, t\hbar Q_{22}+t^2\frac{\hbar^2\sigma^2}{4}Q_{22}
\end{equation}
where $x\cdot E_p x=p^2$, and we can apply \eqref{eq:int-tr-GQ} with $Q=E_p$ to obtain 
\begin{equation}
\tr[ \hat\rho'(0)\ln \hat \rho(0)]=-t\frac{\hbar^2}{2} Q_{22}G_{11}+t^2\frac{\hbar^2\sigma^2}{4}Q_{22}
\end{equation}
and from \eqref{eq:Q-G} and Cramer's rule we derive
\begin{equation}
Q_{22}=2\frac{z}{\hbar}\coth^{-1}(z)\, \frac{G_{11}}{z^2}
\end{equation}
and hence
\begin{equation}
\tr[ \hat\rho'(0)\ln \hat \rho(0)]=-\bigg(tG_{11}^2-\frac{1}{2}t^2 \sigma^2 G_{11}\bigg)\frac{\hbar}{z}\coth^{-1}(z)\,\, .
\end{equation}
From \eqref{eq:G-sc} we see that $G_{11}=g+2\sigma^2t$, and combining the results we find from Lemma \ref{lem:S-FH} 
for the entropy
\begin{equation}\label{eq:entropy-t-approx}
    S(\hat\rho(t))=f(z)+\frac{\gamma^2\hbar t}{2} (g+2\sigma^2 t)\bigg(g+\frac{3}{2}\sigma^2 t\bigg) \frac{1}{z}\ln\bigg(\frac{z+1}{z-1}\bigg)+O(\gamma^4)\,\, ,
\end{equation}
where $z=\sqrt{\det G_t}$, $f(z)$ is given by \eqref{eq:S-quadratic} and we have used $\mathrm{coth}^{-1}(z) = \frac{1}{2}\mathrm{ln}\left( \frac{z+1}{z-1} \right)$. 

\begin{figure}[ht]
\includegraphics[width=0.45\linewidth]{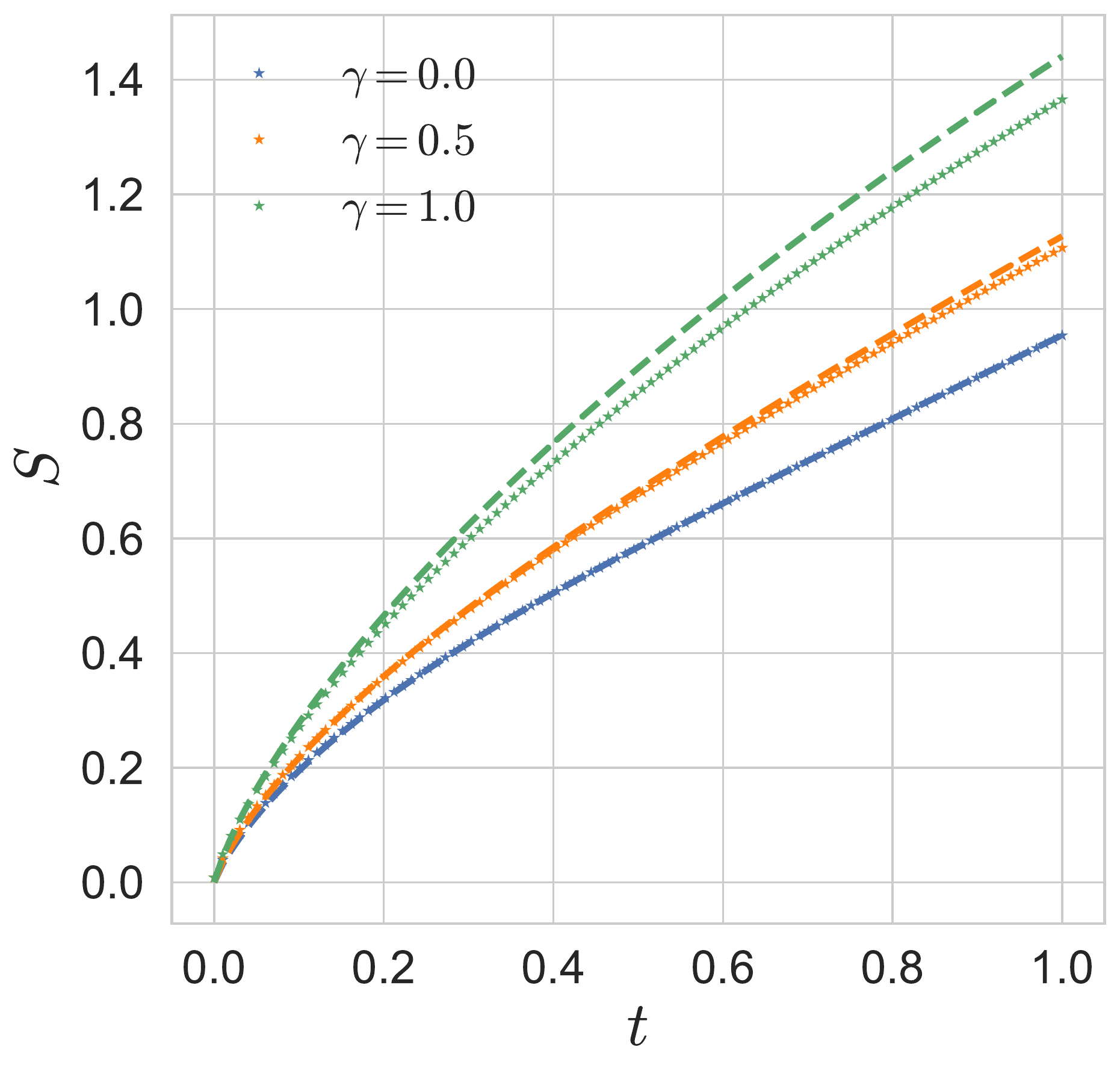}
\includegraphics[width=0.45\linewidth]{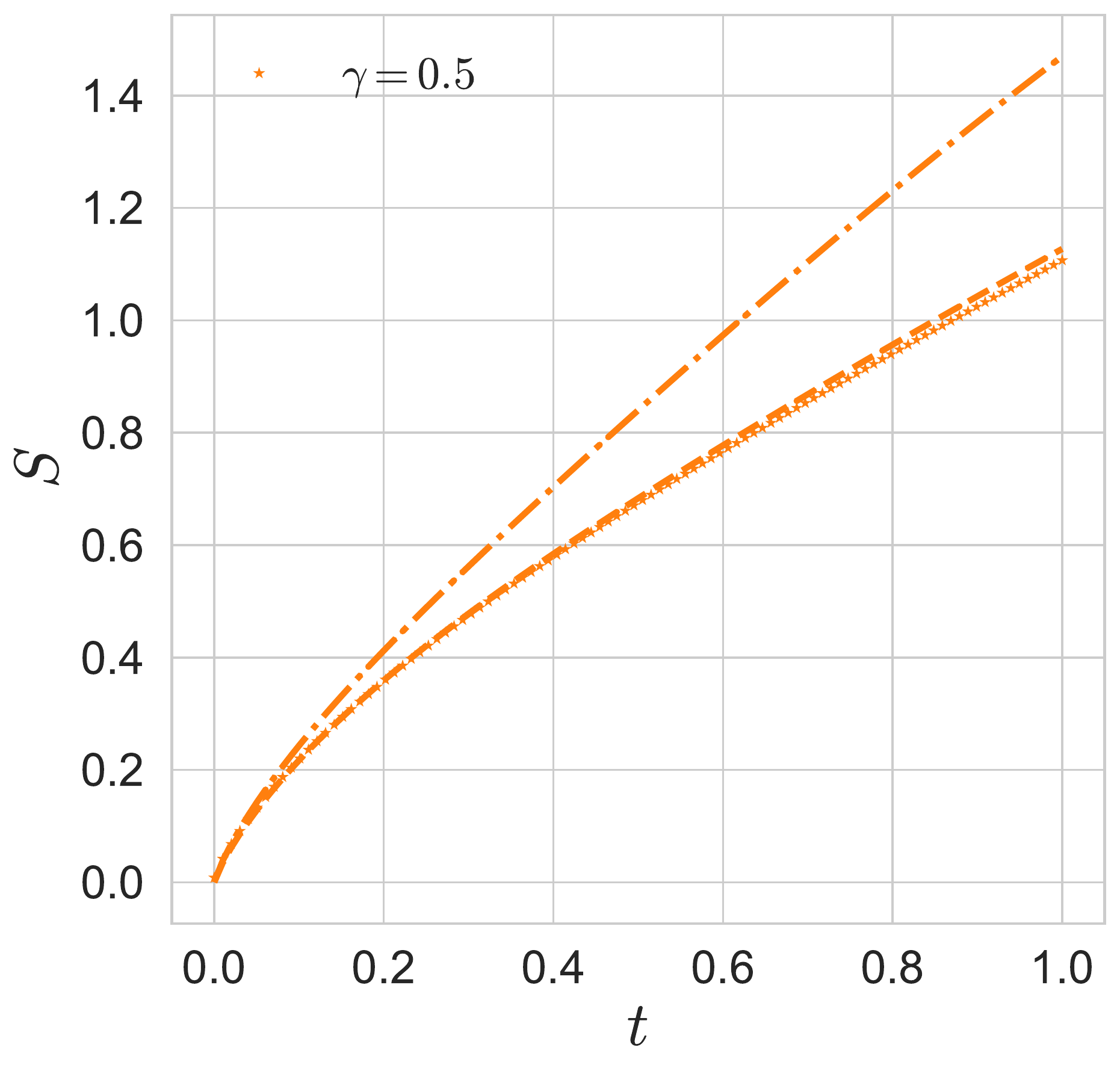}
\caption{Evolution of the entropy $S(\hat\rho(t))$ for an initial 
pure Gaussian state \eqref{eq:Gaussian-0} with $p_0=q_0=0$, $g=1$, and for $\sigma=1$, $\hbar=1$ and different values of $\gamma$. Left panel: The evolution of the entropy for $\gamma=0$, $\gamma=0.5$ and $\gamma=1$ compared with the approximation \eqref{eq:S-cov} (dashed lines). Right panel: The evolution of the entropy for $\gamma=0.5$ compared with the approximations \eqref{eq:S-cov} (dashed line) and \eqref{eq:entropy-t-approx} (dash-dotted line)}\label{fig:entropy}
\end{figure}

In Section \ref{sec:Gaussian-prop}  we were able to compute the covariance $\Gamma$ of $\hat \rho(t)$ explicitly, and we saw that it contains a higher order term which is not covered by the leading order semiclassical approximation \eqref{eq:G-sc}. So it is interesting to compare the entropy of the state  with the entropy of a Gaussian state which has the same covariance, 
\begin{equation}\label{eq:S-cov}
 S_{cov}:=f\big(2\sqrt{\det \Gamma}/\hbar)\big)\,\, ,
\end{equation}
where $f$ is given by \eqref{eq:S-quadratic} and $\Gamma$ by \eqref{eq:Gamma-vs-sc}. 

In Figure \ref{fig:entropy} we computed the evolution of the entropy for an initial Gaussian state \eqref{eq:Gaussian-0} centered at $p_0=q_0=0$ with $\hbar=1$ and $g=1$ for $\sigma=1$ and three different values of $\gamma$: $\gamma=0$, $\gamma=0.5$ and $\gamma=1$. For $\gamma=0$ the evolution is Gaussian and we see the entropy matches \eqref{eq:S-quadratic}. For $\gamma=0.5$ and $\gamma=1$ the entropy increases as one would expect, and we can compare it to the various approximations we have discussed. For an initial state centred at $p_0=0$ the semiclassical variance \eqref{eq:G-sc} does not depend on $\gamma$ and hence using the semiclassical approximation for the state gives an identical entropy for all values of $\gamma$, but we see in Figure \ref{fig:entropy} that this is not accurate. This is not surprising since we have chosen $\hbar=1$ and therefore cannot expect the semiclassical approximation to be very accurate for this parameter value. 
We found that a better approximation is given by a Gaussian state with the same variance as the exact state, \eqref{eq:Gamma-vs-sc}, we compared the entropies \eqref{eq:S-cov} for different $\gamma$ values to the exact entropies on the left panel in Figure \ref{fig:entropy} and found very good agreement. Since the $\hbar$ corrections are explicit in \eqref{eq:Gamma-vs-sc} we see as well how the entropies converge to the semiclassical value for $\hbar\to 0$. 

On the left panel of Figure \ref{fig:entropy} we finally compare the perturbative result \eqref{eq:entropy-t-approx} for the entropy to the exact values, and we see that the agreement is good for small $t$ but then starts deviating and is not as good as the approximation based on the modified variance \eqref{eq:S-cov}. We believe that this is due to the fact that functional form of \eqref{eq:S-cov} reflects the properties of the entropy better and is therefore a better approximation over a wider time range. 



\section{Summary and Outlook}

We considered a family of non-Gaussian quantum channels generated by the Lindblad equation with a free internal Hamiltonian $\hat{H}=\frac{1}{2}\hat p^2$ and Lindblad operators $\hat{L}_1=\sigma \hat q$ and $\hat{L}_2=\gamma \hat{H}$. This is a natural generalisation of the well known Gaussian case $\gamma=0$ which has been extensively used to study decoherence induced by random scattering from an environment, see, e.g., \cite{JoosEtAl03}. We view this as a model case for Quantum Channels beyond the Gaussian case. 

We obtained an explicit expression for the characteristic function of a state obtained by acting with this non-Gaussian Channel on a Gaussian state, and which allowed us in turn to give explicit expressions for expectation values and variances of position and momentum. More generally we derived an explicit expression for the propagator, i.e., the integral kernel of the Quantum Channel. 

We compared these results with the Gaussian case and with several approximations, in particular the semiclassical approximation and a perturbation theory in $\gamma$, and studied how accurately they can describe the evolution of the entropy. Of particular interest is the case of an initial Gaussian state with momentum $0$ since in this case the leading order semiclassical approximation does not detect the dephasing caused by $L_2=\gamma H$, which our analysis shows is an order $\hbar$ correction to the variance. 

The chosen model allows us to derive explicit solutions which we were able to study in quite some detail, but for more general systems we need more qualitative methods. We noticed that our model is related to the Grushin operator, which is a hypoelliptic operator related to sub-Riemannian geometry and analysis on nilpotent Lie groups. We believe that this connection should be explored further and expect that ideas and techniques currently developed in sub-Riemannian geometry, \cite{Agr20}, will prove very useful in the study of the Lindblad equation and open quantum systems. 

\vspace{1cm}
%
\ack
D.S. acknowledges support from the EPSRC Doctoral Prize Fellowship scheme. R.S. acknowledges the financial support provided by the
EPSRC Grant No. EP/P021123/1.

%
%

\section*{Appendices}
\appendix

\section{Some tools from Semiclassical Analysis}\label{app:semi}

In this appendix we recall and summarise a few results from Semiclassical Analysis which we need in the main text, see \cite{Zwo12} for more background.  If $A(p,q)$, $p,q\in\R^n\oplus\R^n$, is a function on phase space, then its Weyl quantisation is an operator $\hat A$ defined as 
\begin{equation}
\big(\hat A\psi\big)(q)=\int K_A(q,q')\psi(q')\, \ud q'\,\, ,\quad\text{with}\quad K_A(q,q')=\frac{1}{(2\pi\hbar)^n}\int\ue^{\frac{\ui}{\hbar}(q-q')p}A\bigg(p,\frac{q+q'}{2}\bigg)\, \ud p\,\, ,
\end{equation}
and the function $A(p,q)$ is called the Weyl symbol of the operator $\hat A$. It satisfies in particular 
\begin{equation}\label{eq:trace-symbol}
\tr \hat A=\frac{1}{(2\pi\hbar)^n}\iint A(p,q)\, \ud p\ud q\,\, ,
\end{equation}
if $A$ is of trace class. Any linear operator has a Weyl symbol, which in general is a distribution. If $\hat\rho$ is a density operator, then the Weyl symbol $\rho(x)$, $x=(p,q)$, and the Wigner function $W_{\hat\rho}(x)$ of $\hat\rho$ are related by 
\begin{equation}
W_{\hat \rho}(x)=\frac{1}{(2\pi\hbar )^n}\rho(x) 
\end{equation}
so that $\int W_{\hat\rho}(x)\ud x=\tr \hat \rho=1$. 

The Weyl symbol of the product of two operators $\hat A$ and $\hat B$ is given by 
\begin{equation}\label{eq:product-AB}
A\#B(x)=\sum_{n=0}^{\infty}\frac{(\ui\hbar)^n}{2^nn!}A(x)(\overleftarrow{\nabla} \cdot\Omega \overrightarrow{\nabla})^nB(x)\,\, ,
\end{equation}
where the sum has to be in general understood as an asymptotic expansion in powers of $\hbar$, but if either $A$ or $B$ is a polynomial, then the sum is finite and the result is exact. 
Here $\Omega$ is the antisymmetric matrix \eqref{eq:Omega} defining the symplectic structure on phase space. The product \eqref{eq:product-AB} is invariant under certain phase space transformations, let $\varphi(x):=Sx+v$, where $v\in \R^n\oplus\R^n$ and $S:\R^n\oplus\R^n\to\R^n\oplus\R^n$ is symplectic, i.e, $S^T\Omega S=\Omega$, then 
\begin{equation}
A_{\varphi}\#B_{\varphi}=\big(A\#B\big)\circ \varphi\,\, ,\quad \text{where}\quad A_{\varphi}:=A\circ\varphi\,\, ,\quad B_{\varphi}=B\circ \varphi\,\, . 
\end{equation}

One nice property of the Weyl calculus is that the trace of the product of two operators has a simple expression in terms of the symbols, 
\begin{equation}\label{eq:prod-tr}
\tr\big(\hat A\hat B\big)=\frac{1}{(2\pi\hbar)^n}\int A(x)B(x)\, \ud x\,\, ,
\end{equation}
the higher order terms in the product formula do not contribute. We use this formula in particular for the case that $A(x)=2^n(\det G)^{-n}\ue^{-\frac{1}{\hbar} x\cdot G^{-1}x}$ and 
$B=\frac{1}{2} x\cdot Qx$, where $G$ and $Q$ are real positive symmetric matrices, and find
\begin{equation}\label{eq:int-tr-GQ}
\tr\big(\hat A\hat B\big)=\frac{1}{(\pi \hbar)^n(\det G)^n}\int \ue^{-\frac{1}{\hbar} x\cdot G^{-1}x}\frac{1}{2} x\cdot Qx\, \ud x=\frac{\hbar}{2} \tr\big(GQ\big)\,\, .
\end{equation}

One application of this formalism is the computation of the Weyl symbol of $\ue^{-\beta \hat H}$, where $H(x)=\frac{1}{2}x\cdot Qx$ is a positive definite quadratic form. 

\begin{lem}Let $\hat A(\beta)=\ue^{-\beta \hat H}$ with $H=\frac{1}{2}{x\cdot Qx}$ where $Q^T=Q$ and $Q>0$, $\beta >0$, and $x=(p,q) \in \R^2$, then 
\begin{equation}
A(\beta,x)=\frac{1}{\cosh(\omega \hbar \beta/2)}\ue^{-\frac{2\tanh(\omega\hbar \beta/2)}{\omega \hbar}H(x)} \quad\text{and}\quad Z(\beta):=\tr \hat A(\beta)=\frac{1}{2\sinh(\hbar \omega\beta/2)}\,\, ,
\end{equation}
where $\omega =\sqrt{\det Q}$. Furthermore, the state $\rho_{\beta}:=\ue^{-\beta \hat H}/Z(\beta)$ is Gaussian and its characteristic function is given by 
\begin{equation}
\chi_{\rho_{\beta}}(\xi)=\ue^{-\frac{1}{4\hbar}\frac{\omega}{\tanh(\omega\hbar/2)}\xi\cdot Q^{-1}\xi}\,\, .
\end{equation}

\end{lem}

\begin{proof}
There exists a symplectic matrix $S$ such that $H(S(x))=\frac{\omega}{2}x^2$, where $\omega=\sqrt{\det Q}$. 
We will use that $A(\beta)$ satisfies the equation
\begin{equation}\label{eq:A-beta}
\pa_{\beta}A(\beta)=-H\# A(\beta)\,\, \quad A(0)=1\,\, , 
\end{equation}
which for $A_{\varphi}(\beta)$ with $\varphi(x)=S(x)$ becomes 
\begin{equation}\label{eq:A-beta-phi}
\pa_{\beta}A_{\varphi}(\beta)=-H_{\varphi}\# A_{\varphi}(\beta)\,\, \quad A_{\varphi}(0)=1\,\, , \quad H_{\varphi}(x)=\frac{\omega}{2}x^2
\end{equation}
We make an Ansatz as $A_{\varphi}(\beta)=c(\beta)\ue^{-f(\beta)H_{\varphi}(x)}$ and inserting this into \eqref{eq:A-beta-phi} and using \eqref{eq:product-AB} gives 
\begin{equation}
\frac{c'}{c}A_{\varphi}(\beta,x)-f'H_{\varphi}(x)A_{\varphi}(\beta,x)=-H_{\varphi}(x)A_{\varphi}(\beta,x)+\frac{\hbar^2}{8}H_{\varphi}(x)(\overleftarrow{\nabla} \cdot\Omega \overrightarrow{\nabla})^2A_{\varphi}(\beta,x)
\end{equation}
and with 
\begin{equation}
H_{\varphi}(x)(\overleftarrow{\nabla} \cdot\Omega \overrightarrow{\nabla})^2A_{\varphi}(\beta,x)=\omega(\pa_p^2+\pa_q^2)A_{\varphi}(\beta ,x)=(-2\omega^2 f+2\omega^2 f H_{\varphi}(x))A_{\varphi}(\beta,x)
\end{equation}
we obtain the two equations
\begin{equation}
f'=1-\frac{\hbar^2\omega^2}{4}f\,\, , \quad f(0)=0
\end{equation}
and
\begin{equation}
 \frac{c'}{c}=-\frac{\hbar^2\omega^2}{4} f\,\, \quad c(0)=1
 \end{equation}
 The solutions of these equations are 
 \begin{equation}
 f(\beta)=\frac{2\tanh(\hbar\omega\beta/2)}{\hbar \omega}\,\, ,\quad c(\beta)=\frac{1}{\cosh(\hbar\omega\beta/2)}\,\, 
 \end{equation}
and if we compose $A_{\varphi}(\beta)$ with $\varphi^{-1}$ we get the formula for $A(\beta)$. The expression for the trace follows by using \eqref{eq:trace-symbol} and the formula for the characteristic function follows by computing the Fourier transform of $A(\beta,x)/\tr \hat A(\beta)$.  
\end{proof}

One of the applications of this lemma is that it allows us to compute the logarithm of a Gaussian state \eqref{eq:Gaussian-state-diag-G}. Let $\hat \rho$ be a Gaussian state with covariance matrix $G$, satisfying $G+\ui\Omega\geq 0$, then by comparing the characteristic functions (for $\beta=1$) we see that 
\begin{equation}
\hat \rho=\frac{1}{Z(1)}\ue^{- \hat H}
\end{equation}
if 
\begin{equation}\label{eq:G-Q-1}
G=\frac{\omega}{\tanh(\omega\hbar/2)}\, Q^{-1}\,\, ,
\end{equation} 
where $H=\frac{1}{2}x\cdot Qx$ and $\omega=\sqrt{\det Q}$. This gives us then the expression 
\begin{equation}
\ln\hat\rho=- \hat H-\ln Z(1)\,\, .
\end{equation}
Taking the determinant of \eqref{eq:G-Q-1} leads to 
\begin{equation}
z=\frac{1}{\tanh(\omega\beta \hbar/2)}\,\, ,\quad\text{where}\quad z=\sqrt{\det G}
\end{equation}
and then we find
\begin{equation}
Q=2\frac{z}{\hbar}\coth^{-1}(z)\, G^{-1}\,\, ,
\end{equation}
and 
\begin{equation}
Z(1)=\frac{1}{2}\sqrt{z^2-1}\,\, .
\end{equation}


\section{Wave Packet Propagation}\label{app:wpp}

In this appendix we give the details on some of the more technical calculations involving semiclassical analysis which are used in Section \ref{sec:Gaussian-prop}. 

In the first part we describe how to derive the the wave packet propagation for Equation \eqref{eq:w_Lindblad} for Gaussian wave packets. These results can be derived as a special case of the non-Hermitian propagation in \cite{GraeSchu11,GraeSchu12}, but for sake of clarity we include an explicit derivation here.

The equation \eqref{eq:w_Lindblad} reads 
\begin{equation}\label{eq:w_Lindblad_2}
\hbar \pa_tw=\frac{\hbar^2\sigma^2}{2} \pa_p^2 w-\frac{\gamma^2\eta^2}{2}\, p^2 w+\ui\eta p w \,\, .
\end{equation}
and we make an Ansatz 
\begin{equation}\label{eq:w_ansatz}
w_t(p)=c(t)\ue^{-\frac{\alpha(t)}{\hbar}}\ue^{-\frac{\ui}{\hbar} Q(t) p}\ue^{-\frac{1}{\hbar}\frac{1}{a(t)}(p-P(t))^2}
\end{equation}
where the time dependent parameters $c(t), a(t), Q(t),P(t)$ are real and $a(t)>0$, and $\alpha(t)=D(t)-\ui\phi(t)$ is complex with $\Re \alpha (t)=D(t)\geq 0$. 

We will insert the Ansatz \eqref{eq:w_ansatz} into \eqref{eq:w_Lindblad_2} and then ordering by powers of $(p-P)$ will give us a set of equations for the parameters. For the left hand side of \eqref{eq:w_Lindblad_2} we find
\begin{equation}
\hbar \pa_tw=\bigg[\hbar \frac{\dot c}{c}-\dot\alpha -\ui \dot Q (p-P)-\ui \dot Q P+\frac{\dot a}{a^2}(p-P)^2+\frac{2}{a} \dot P(p-P)\bigg] w
\end{equation}
and for the right hand side we first compute 
\begin{equation}
\frac{\hbar^2\sigma^2}{2} \pa_p^2 w=\bigg(\frac{4}{a^2}(p-P)^2+\ui\frac{4}{a}Q(p-P)-Q^2-\frac{2\hbar}{a}\bigg)w
\end{equation}
and with $p=p-P+P$ and $p^2=(p-P)^2+2P(p-P)+P^2$ we get for the right hand side of \eqref{eq:w_Lindblad_2}
\begin{equation}
\begin{split}\bigg[\frac{2\sigma^2}{a^2}(p-P)^2&+\ui\frac{2\sigma^2}{a} Q(p-P)-\frac{\sigma^2}{2}Q^2-\frac{\sigma^2\hbar}{a}-\frac{\gamma^2\eta^2}{2}(p-P)^2\\&-\gamma^2\eta^2 P(p-P)-\frac{\gamma^2\eta^2}{2}P^2+\ui \eta (p-P)+\ui \eta P\bigg]w\,\, .
\end{split}
\end{equation}

If we now compare different powers of $(p-P)$ we get from the term containing $(p-P)^2$ an equation for $a$ 
\begin{equation}
\dot a=2\sigma^2-\frac{\gamma^2\eta^2}{2} a^2\label{eq:a}\, \, . 
\end{equation}
The terms containing $(p-P)$ give us an equation for $P$ and $Q$, where separating the real and imaginary parts leads to
\begin{align}
\dot P&=-\frac{\gamma^2\eta^2}{2} a P \label{eq:P}\\
\dot Q&=-\frac{2\sigma^2}{a}Q-\eta\,\, .
\end{align}
Finally for $\alpha$ and $c$ we find
\begin{align}
\dot c&=-\frac{\sigma^2}{a} c\\
\dot \alpha&=\frac{\sigma^2}{2}Q^2+\frac{\gamma^2\eta^2}{2}P^2-\ui \dot Q P-\ui \eta P\,\, .
\end{align}
Splitting the equation for $\alpha$ into its real and imaginary parts gives the equations \eqref{eq:D} and \eqref{eq:phi} for $D$ and $\phi$. 

In order to find the solutions to these equations it is useful to introduce $u(t)=\frac{P_0}{P(t)}$ and by taking derivatives of $u$ and  using \eqref{eq:P} and \eqref{eq:a} we find that 
\begin{equation}\label{eq:f}
a=\frac{2}{\gamma^2\eta^2}\frac{\dot u}{u} \,\, , \quad \text{and}\quad \ddot{u}=\gamma^2\sigma^2\eta^2 u\,\, .
\end{equation}
The second equation immediately gives $u(t)=u_0\cosh(\omega t)+u_1\sinh(\omega t)$, where $\omega=\sigma \gamma\eta$, and then the first equation reads $a(t)=\frac{2\sigma}{\gamma\abs{\eta}}\frac{u_0\sinh(\omega t)+u_1\cosh(\omega t)}{u_0\cosh(\omega t)+u_1\sinh(\omega t)}$. From the initial conditions $u(0)=u_0=1$ and $a(0)=a_0=g$ we get the values of $u_0$ and $u_1$ and find
\begin{equation}\label{eq:a-P}
a(t)=g\frac{v(t)}{u(t)}\,\, ,\quad 
P(t)=\frac{P_0}{u(t)}\,\, .
\end{equation}
where we introduced the auxiliary functions 
\begin{equation}\label{eq:def-u-v}
u(t)=\cosh(\omega t)+\frac{\omega}{\beta}\sinh(\omega t)\, \qquad 
v(t)=\cosh(\omega t)+\beta\frac{\sinh(\omega t)}{\omega}
\end{equation}
with 
\begin{equation}
\omega =\sigma \gamma \eta\,\, \quad \beta =\frac{2\sigma^2}{g}\,\, .
\end{equation}
The two functions are related by 
\begin{equation}\label{eq:u-v-derivatives}
\dot u=\frac{\omega^2}{\beta} v\,\, \quad \text{and}\quad \dot v=\beta u\,\, .
\end{equation}

Notice that the equation for $c$ looks very similar to the one for $P$ and with \eqref{eq:f} we find
\begin{equation}
\frac{\ud}{\ud t} \ln c=\frac{\dot c}{c}=-\frac{\sigma^2}{a}=-\frac{\sigma^2\gamma^2\eta^2}{2}\frac{u}{\dot u}=-\frac{1}{2}\frac{\ddot{u}}{\dot u}=\frac{\ud}{\ud t} \ln {\dot u}^{-1/2}\,\, ,
\end{equation}
and so 
\begin{equation}
c(t)=\frac{c_0}{\sqrt{v(t)}}\,\, .
\end{equation}

To solve the equation for $Q$ we first consider the homogeneous part of the equation
\begin{equation}
\dot{\tilde Q}=-\frac{2\sigma^2}{a} \tilde Q
\end{equation}
and using the same strategy as for $P$ we find that the solution with $\tilde Q(t=0)=1$ is given by $\tilde Q(t)=1/v(t)$. Then the variation of constants formula gives for $Q$ the solution 
\begin{equation}
Q(t)=\frac{1}{v(t)}\bigg[Q_0-\eta \int_0^t v(s)\,\, \ud s\bigg]\,\, ,
\end{equation}
and the integral can be easily computed using \eqref{eq:u-v-derivatives} to give
\begin{equation}\label{eq:Q(s)}
Q(t)=\frac{1}{v(t)}\bigg[ Q_0-\eta \beta\frac{u(t)-1}{\omega^2}\bigg]\,\, .
\end{equation}
Notice that 
\begin{equation}
\frac{u(t)-1}{\omega^2}=\frac{\cosh(\omega t)-1}{\omega^2}+ \frac{\sinh(\omega t)}{ \beta \omega}
\end{equation}
and therefore $Q(t)$ is smooth at $\omega=0$. 

Finally we can compute $D(t)$ and and $\phi(t)$ which are given by 
\begin{equation}
D(t)=\frac{1}{2}\int_0^t \sigma^2 Q(s)^2+\gamma^2\eta^2 P(s)^2\,\, \ud s
\end{equation}
and 
\begin{equation}
\phi(t)=\int_0^t \dot Q(s)P(s)+\eta P(s)\,\, \ud s=-2\sigma^2\int_0^t \frac{Q(s)P(s)}{a(s)}\, \ud s\,\, .
\end{equation}
We collect the more technical integrals we need in the following lemma. 

\begin{lem}\label{lem:integrals} Let $u(t)$ and $v(t)$ be given by \eqref{eq:def-u-v}, then we have
\begin{align}
\int_0^t \frac{u(s)}{v(s)^2}\, \,\ud s&=\frac{1}{\beta}\frac{v(t)-1}{v(t)}\\
\int_0^t\frac{1}{v(s)^2}\,\, \ud s &=\frac{\sinh(\omega t)}{\omega}\frac{1}{v(t)}\\
\int_0^t\frac{1}{u(s)^2}\,\, \ud s &=\frac{\sinh(\omega t)}{\omega}\frac{1}{u(t)}\\
\int_0^t\frac{u(s)^2}{v(s)^2}\,\, \ud s &=\frac{\omega^2}{\beta^2} \, t+\bigg(1-\frac{\omega^2}{\beta^2}\bigg)\frac{\sinh(\omega t)}{\omega}\frac{1}{v(t)}\,\, .
\end{align}

\end{lem}

\begin{proof}
We have by \eqref{eq:u-v-derivatives} that $-\beta u/v^2=\frac{\ud}{\ud s}\frac{1}{v}$ and hence 
\begin{equation}
\int_0^t \frac{u}{v^2}\, \ud s=-\frac{1}{\beta}\bigg[\frac{1}{v(t)}-\frac{1}{v(0)}\bigg]=\frac{1}{\beta}\frac{v(t)-1}{v(t)}\,\, .
\end{equation}
In preparation for the remaining integrals we notice that by \eqref{eq:u-v-derivatives}
\begin{equation}\label{eq:Wronskian}
\dot v u-v\dot u=\beta u^2-\frac{\omega^2}{\beta}v^2=\beta -\frac{\omega^2}{\beta}
\end{equation}
where the last identity follows by noticing the expression is the Wronskian of $v$ and $u$, and hence constant, and therefore equal to its value  at $t=0$.
Returning to the integrals we find using \eqref{eq:Wronskian}
\begin{align}
\frac{\ud}{\ud s}\frac{v}{u}&=\frac{\dot v u-v\dot u}{u^2}=\frac{\beta^2-\omega^2}{\beta}\frac{1}{u^2}\\
\frac{\ud}{\ud s}\frac{u}{v}&=\frac{\dot u v-u\dot v}{v^2}=-\frac{\beta^2-\omega^2}{\beta}\frac{1}{v^2}
\end{align}
and integrating the second relation and using $v(t)-u(t)=\frac{\beta^2-\omega^2}{\beta\omega}\sinh(\omega t)$ gives 
\begin{equation}
\int_0^t\frac{1}{v^2}\, \ud s=-\frac{\beta}{\beta^2-\omega^2}\bigg(\frac{u(t)}{v(t)}-1\bigg)=\frac{\sinh(\omega t)}{\omega}\frac{1}{v(t)}\,\, ,
\end{equation}
and similarly we obtain
\begin{equation}
\int_0^t\frac{1}{u^2}\, \ud s=\frac{\sinh(\omega t)}{\omega}\frac{1}{u(t)}\,\, .
\end{equation}
Finally, from \eqref{eq:Wronskian} we get 
\begin{equation}
\frac{u^2}{v^2}=\frac{\omega^2}{\beta^2}+\bigg(1-\frac{\omega^2}{\beta^2}\bigg)\frac{1}{v^2}
\end{equation}
and integrating this relation gives 
\begin{equation}
\int_0^t\frac{u^2}{v^2}\,\, \ud s=\frac{\omega^2}{\beta^2} \, t+\bigg(1-\frac{\omega^2}{\beta^2}\bigg)\frac{\sinh(\omega t)}{\omega}\frac{1}{v(t)}\,\, .
\end{equation}
\end{proof}

\begin{lem} We have
\begin{equation}\label{eq:int-Q}
\begin{split}
\int_0^tQ(s)^2\, \ud s=&\frac{\eta^2}{\omega^2}\bigg(t v(t,\omega)-\frac{\sinh(\omega t)}{\omega}-2\beta \frac{\cosh(\omega t)-1}{\omega^2}\bigg)\frac{1}{v(t,\omega)}\\
&-2\eta Q_0\frac{\cosh(\omega t)-1}{\omega^2}\frac{1}{v(t,\omega)}+Q_0^2\frac{\sinh(\omega t)}{\omega}\frac{1}{v(t,\omega)}\,\, , 
\end{split}
\end{equation}
\begin{equation}\label{P_sqr_integral}
\int_0^tP(s)^2\, \ud s=P_0^2\frac{\sinh(\omega t)}{\omega}\frac{1}{u(t,\omega)}
\end{equation}
and 
\begin{equation}
\int_0^t \frac{Q(s)P(s)}{a(s)}\, \ud s=\frac{P_0Q_0}{g}\frac{\sinh(\omega t)}{\omega}\frac{1}{v(t)}-\frac{P_0\eta }{g}\frac{\cosh(\omega t)-1}{\omega^2}\frac{1}{v(t)}\,\, .
\end{equation}
\end{lem}

\begin{proof}
Using \eqref{eq:Q(s)} we have 
\begin{equation}
\int_0^t Q^2\, \ud s =Q_0^2\int_0^t \frac{1}{v^2}\, \ud s-2Q_0\frac{\eta \beta}{\omega^2}\int_0^t\frac{u-1}{v^2}\, \ud s+\frac{\eta^2\beta^2}{\omega^4}\int_0^t \frac{(u-1)^2}{v^2}\, \ud s
\end{equation}
and the individual integrals give 
\begin{align}
\int_0^t\frac{1}{v^2}\, \ud s &=\frac{\sinh(\omega t)}{\omega}\frac{1}{v(t)}\,\, ,\\
\int_0^t\frac{u-1}{v^2}\,\, \ud s &=\frac{1}{\beta}\frac{v-1}{v}-\frac{\sinh(\omega t)}{\omega}\frac{1}{v}=\frac{1}{\beta}\frac{\cosh(\omega t)-1}{v(t)}
\end{align}
and
\begin{equation}
\begin{split}
\int_0^t \frac{(u-1)^2}{v^2}\, \ud s&=\int_0^s \frac{u^2}{v^2}-2\frac{u}{v^2}+\frac{1}{v^2}\, \ud s\\
&=\frac{\omega^2}{\beta^2} \, t +\big(1-\frac{\omega^2}{\beta^2}\bigg)\frac{\sinh(\omega t)}{\omega}\frac{1}{v}-2\frac{1}{\beta}\frac{v-1}{v}+\frac{\sinh(\omega t)}{\omega}\frac{1}{v}\\
&=\frac{\omega^2}{\beta^2}\bigg(tv(t)-\frac{\sinh(\omega t)}{\omega}-2\beta \frac{\cosh(\omega t)-1}{\omega^2}\bigg)\frac{1}{v}\,\, .
\end{split}
\end{equation}
and combining these gives \eqref{eq:int-Q}.

The integral of $P(s)^2=P_0^2/u(s)^2$ given by equation \eqref{P_sqr_integral} follows directly from Lemma \ref{lem:integrals}. 

For the final integral we use that by \eqref{eq:a-P} we have $P(s)/a(s)=P_0/(gv(s))$ and hence 
\begin{equation}
\begin{split}
\int_0^t \frac{Q(s)P(s)}{a(s)}\, \ud s=&\frac{P_0}{g}\int_0^t \frac{Q(s)}{v(s)}\, \ud s\\
=&\frac{P_0}{g}\bigg(Q_0+\frac{\eta\beta}{\omega^2}\bigg) \int_0^t \frac{1}{v^2}\, \ud s-\frac{P_0\eta \beta}{g \omega^2}\int_0^t \frac{u}{v^2}\, \ud s\\
=&\frac{P_0}{g}\bigg(Q_0+\frac{\eta\beta}{\omega^2}\bigg) \frac{\sinh(\omega t)}{\omega}\frac{1}{v}-\frac{P_0\eta \beta}{g \omega^2}\frac{1}{\beta}\frac{v-1}{v}\\
=&\frac{P_0Q_0}{g}\frac{\sinh(\omega t)}{\omega}\frac{1}{v(t)}-\frac{P_0\eta }{g}\frac{\cosh(\omega t)-1}{\omega^2}\frac{1}{v(t)}\,\, .
\end{split}
\end{equation}

\end{proof}

Using these results we find that 
\begin{equation}
\phi(t,\eta)=-\frac{2\sigma^2P_0}{g}\bigg[Q_0\frac{\sinh(\omega t)}{\omega}-\eta \frac{\cosh(\omega t)-1}{\omega^2}\bigg]\frac{1}{v(t,\omega)}
\end{equation}
and 
\begin{equation}
\begin{split}
D(t,\eta)=&\frac{1}{2\gamma^2}\bigg(t v(t,\omega)-\frac{\sinh(\omega t)}{\omega}-2\beta \frac{\cosh(\omega t)-1}{\omega^2}\bigg)\frac{1}{v(t,\omega)}\\
&-\sigma^2\eta Q_0\frac{\cosh(\omega t)-1}{\omega^2}\frac{1}{v(t,\omega)}+\frac{\sigma^2Q_0^2}{2}\frac{\sinh(\omega t)}{\omega}\frac{1}{v(t,\omega)}+\frac{\gamma^2\eta^2P_0^2}{2}\frac{\sinh(\omega t)}{\omega}\frac{1}{u(t,\omega)}
\end{split}
\end{equation}
where $\omega=\sigma\gamma\eta$. We have arranged the terms so that the limits $\omega\to 0$ and $\gamma\to 0$ do not cause any artificial singularities, in particular we have 
\begin{equation}
t v(t,\omega)-\frac{\sinh(\omega t)}{\omega}-2\beta \frac{\cosh(\omega t)-1}{\omega^2}=\omega^2\bigg(\frac{1}{3}t^3+\frac{\beta}{12}t^4\bigg)+O(\omega^3)
\end{equation}
which implies that the limit $\gamma\to 0$ of $D(t, \eta)$ is well defined and gives 
\begin{equation}
\lim_{\gamma\to 0}D(t, \eta)=\frac{\sigma^2}{1+\beta t}\bigg(\frac{1}{6}t^3+\frac{\sigma^2}{12 g}t^4\bigg)\eta^2-\frac{Q_0\sigma^2}{2}\frac{t^2}{1+\beta t}\, \eta+\frac{Q_0^2\sigma^2}{2}\frac{t}{1+\beta t}\,\, .
\end{equation}


\section*{References}
\bibliographystyle{amsalpha}
\bibliography{BGCbibliography}

\end{document}